\documentclass{lmcs}
\usepackage[utf8]{inputenc}
\usepackage{textcomp}
\usepackage{amssymb,bm,graphicx,array}
\usepackage{stmaryrd}
\SetSymbolFont{stmry}{bold}{U}{stmry}{m}{n}
\usepackage{hyperref}
\usetikzlibrary{arrows.meta,positioning,decorations.pathreplacing}

\newcommand{\Nat}{\mathbb{N}}
\newcommand{\Natp}{\mathbb{N}^{+}}
\newcommand{\Int}{\mathbb{Z}}
\newcommand{\Rat}{\mathbb{Q}}
\newcommand{\Real}{\mathbb{R}}
\newcommand{\Ratp}{\mathbb{Q}_{\geq 0}}
\newcommand{\Realp}{\mathbb{R}_{\geq 0}}
\newcommand{\EL}{\mathsf{EL}}
\newcommand{\ELN}{\mathsf{EL}(\mathbb{N})}
\newcommand{\ELQ}{\mathsf{EL}(\mathbb{Q})}
\newcommand{\SM}{\mathsf{EL}_{\lambda}}
\newcommand{\LTL}{\mathsf{LTL}}
\newcommand{\MTL}{\mathsf{MTL}}
\newcommand{\FO}{\mathsf{FO}}
\newcommand{\Prop}{\mathsf{P}}
\newcommand{\IE}{\mathsf{IE}}
\newcommand{\boldBox}{\Box}
\newcommand{\mdiamond}{\Diamond}
\newcommand{\md}{\Diamond^{\to}}
\newcommand{\mb}{\Box^{\to}}
\newcommand{\sem}[1]{\llbracket #1\rrbracket}
\newcommand{\norm}[1]{\lVert #1\rVert}
\newcommand{\Lang}{\mathcal{L}}
\newcommand{\cent}{\mbox{\textcent}}
\newcommand{\oReg}{\omega\text{-}\mathrm{Reg}}
\newcommand{\cl}[2]{\genfrac{}{}{0pt}{1}{#1}{#2}}
\newcommand{\cL}[2]{{\renewcommand{\arraystretch}{0.92}\begin{array}{@{}c@{}}\smash{#1}\\\smash{#2}\end{array}}}%
\newcommand{\mk}[1]{\textcolor{inkred}{\dot{\bm{#1}}}}%
\newcommand{\Tm}{\mathsf{T}}
\newcommand{\Bm}{\mathsf{B}}
\newcommand{\Init}{\mathsf{Init}}
\newcommand{\Frame}{\mathsf{Frame}}
\newcommand{\Copy}{\mathsf{Copy}}
\newcommand{\Step}{\mathsf{Step}}
\newcommand{\Last}{\mathsf{Last}}
\newcommand{\One}{\mathsf{One}}
\newcommand{\Grid}{\mathsf{Grid}}
\newcommand{\Diag}{\mathsf{Diag}}
\newcommand{\Blk}{\mathsf{Blk}}
\newcommand{\Uniq}{\mathsf{Uniq}}
\newcommand{\X}{\mathsf{X}}
\newcommand{\F}{\mathsf{F}}
\newcommand{\G}{\mathsf{G}}
\newcommand{\U}{\mathbin{\mathsf{U}}}
\newcommand{\Us}{\mathbin{\mathsf{U}^{\circ}}}
\newcommand{\Ev}{\mathsf{E}}
\definecolor{inkblue}{RGB}{31,78,140}
\definecolor{inkred}{RGB}{160,40,40}
\definecolor{inkgreen}{RGB}{30,110,60}

\hypersetup{pdfsubject={Ensemble Logic -- Version 18; LMCS manuscript}}

\begin{document}

\title[Ensemble Logic]{Ensemble Logic}
\author[G.-Q.~Zhang]{Guo-Qiang Zhang}
\address{The University of Texas Health Science Center at Houston, Houston, Texas, USA}
\email{guo-qiang.zhang@uth.tmc.edu}
\ACMCCS{Theory of computation $\rightarrow$ Modal and temporal logics;
Logic and verification}
\date{}
\begin{abstract}
Ensemble Logic (EL) is a family of modal logics in which the parameters of
the modal operators are terms with quantified variables. Formulas are
interpreted over metric monoids: commutative monoids equipped with
translation-invariant metrics. These provide a common framework for temporal
and spatial reasoning in discrete, dense and multi-dimensional settings.
Displacement evaluates a formula at a translated point, while bounded
existence and bounded universality evaluate it at some or every point of a
window of prescribed length. Quantifiers range over points and lengths, and
one variable may occur in the parameters of several operators.
Our main result is a sharp undecidability boundary for satisfiability over
the natural numbers: satisfiability is undecidable
already for formulas $\exists N\,\exists n\,\Psi$ in which $\Psi$ is built
from the letters of a fixed alphabet by conjunction, disjunction,
displacement and bounded universality. The proof reduces Post's
Correspondence Problem. One shared length determines the row width of a table
encoded as an $\omega$-word. Displacement relates consecutive rows,
and a window bounded by the second length imposes the required conditions
throughout the table. The boundary is sharp:
satisfiability becomes decidable when one of the two quantifiers or any one
of the four constructs is omitted. The reduction also shows that validity is
not recursively enumerable, even for quantifier-free formulas. The
undecidability extends to the linear-time logics of the family; in dense
time, two shared lengths again suffice. A single universal quantifier, without
bounded modalities, also yields undecidability. Further results, in
part from companion work, place this theorem within the general theory of EL:
an atomic-displacement normal form for covariant windows, and
composition laws for discrete and dense time. The discrete logic admits an
embedding of linear temporal logic, defines the accepting runs of every
B\"uchi automaton, and has a sound equational calculus. For the rational
logic, a fixed formula is evaluated in polynomial time on labelled
interval ensembles, which interpret atomic propositions by finite unions of
intervals. We also represent timed words over positions and real time and
compare two window geometries; with rectangular windows, two shared lengths
yield undecidability on finite timed words. This prepares the study of
real-time logics within the family.
\end{abstract}
\maketitle

\section{Introduction}\label{sec:intro}

A temporal logic describes a structure from a point of
evaluation. What it can express depends largely on how its operators
reach other points. Linear temporal
logic~\cite{pnueli1977temporal,pnueli1992temporal} follows the order of
positions by means of next and until. Metric temporal logic attaches constant
intervals to its operators~\cite{metric,alurh93}, timed propositional temporal
logic freezes clock readings~\cite{alur-henzinger}, and parametric temporal
logic admits variable bounds~\cite{aetp2001}. In \emph{Ensemble
Logic} ($\EL$), the parameter of a modal operator is a term whose variables
can be bound by first-order quantifiers. A formula can therefore quantify over
displacements and window lengths, and can use one quantified variable in
several parameters. Much of this paper concerns what such sharing brings, and
what it costs.

We interpret formulas over a metric monoid, whose elements we call points.
Discrete time $\Nat$, dense time $\Rat$ or $\Real$, the grids $\Int^k$, the
Euclidean spaces $\Real^k$, and products such as $\Nat\times\Real$ are
examples. An interpretation assigns atomic propositions to points. The
language has three modal operators. Displacement $\varphi_u$ evaluates
$\varphi$ at the point obtained by adding the point denoted by $u$ to the
current point. Bounded existence $\mdiamond_t\varphi$ and bounded
universality $\boldBox_t\varphi$ evaluate $\varphi$ at some or every point,
respectively, of a window. Its points lie at a distance smaller than the
length denoted by $t$ from the point of evaluation. Typical windows are
intervals in linear time and open balls in space. Quantifiers range over
points and lengths. Atomic propositions are the only predicates; terms
are never compared with one another. A metric monoid and a choice of
windows determine a particular logic of the family, while the satisfaction
clauses are common to all its members.

What a shared variable contributes can be seen already over $\Nat$. An
interpretation with exactly one letter at each position is an $\omega$-word.
The window of length $n$ at position $i$ is $\{i,\ldots,i+n-1\}$. For
distinct letters $a,b,c$ and a length variable $n$, the formula
\[
\exists n\,\bigl(\boldBox_n a\wedge(\boldBox_n b)_n\wedge(\boldBox_n c)_{2n}\bigr)
\]
holds at position $0$ of exactly the $\omega$-words that begin with
$a^nb^nc^n$ for some $n\geq1$: the chosen length determines all three windows
and the displacements of the second and third. This set of $\omega$-words is
not an $\omega$-regular language~\cite{buchi62,Tho90}, and is therefore
neither definable in linear temporal logic nor recognizable by a B\"uchi
automaton. Here the shared variable compares the lengths of windows. It can
also compare their contents. When a quantified length specifies the row
width, an $\omega$-word becomes a table. Displacement by that length
leads to the cell in the same column of the next row. The proof
of our main theorem is built on this reading.

This paper develops a \emph{unifying paradigm} for these logics. Its central
result is a sharp undecidability theorem for the
discrete logic $\ELN$. The paradigm consists of a
two-sorted syntax and a satisfaction relation over metric monoids, together
with laws whose validity depends on stated assumptions on the monoid and its
windows. The theorem shows how little sharing $\ELN$ needs for
undecidability. Our
\emph{main theorem} states that satisfiability for $\ELN$ is undecidable
already for closed formulas $\exists N\,\exists n\,\Psi$ over an alphabet of
thirteen letters, where $\Psi$ is built from letters by conjunction,
disjunction, displacement and bounded universality
(Theorem~\ref{thm:main}). Two shared lengths and one of the two bounded
modalities suffice, and neither universal quantification nor negation is
required. The result is sharp: satisfiability is decidable
with one existential quantifier, and with any number of them when one of the
four constructs is omitted. The proof also shows that validity is not
recursively enumerable even for quantifier-free formulas, so that no sound
proof system with a decidable notion of proof is complete for $\ELN$.
The undecidability passes to
the linear-time logics of the family. Two shared lengths again
suffice in dense time and on timed words. The remaining results supply what
the theorem requires and set it in context: the general laws, a
comparison of the discrete logic with classical language classes, and a sound
equational calculus. Polynomial data complexity for model checking and a
representation of timed words complement them.

The name of the logic refers to the support of a formula, the set of points
at which it holds. This is generally an ensemble of separate components.
The modal operators act on it by translation, dilation and erosion, in the
sense of mathematical morphology~\cite{matheron1975,serra1982}. $\EL$ was
originally motivated by the specification of temporal patterns in annotated
biomedical records~\cite{zhang2022temporal,zhang2024temporal}, and its
applications include clinical-trial protocols~\cite{time25,scaleup},
sleep-scoring guidelines~\cite{sleep}, and the identification of patient
cohorts in electrophysiological recordings~\cite{qel,cohortengine}. A record
determines an interpretation and a specification determines a
formula. Answering a query is thus a model-checking
problem~\cite{clarke1997model,vardi2005model}: the formula is evaluated in
the given interpretation. A formal semantics supports the aim
that conclusions drawn from data be transparent, rigorous, understandable,
explainable and reproducible~\cite{qel}. For this reason we treat evaluation
alongside satisfiability.

\subsection{Contributions}\label{subsec:contrib}
The results fall into four areas; those recalled from companion work are
identified below.

\begin{enumerate}[leftmargin=2em]
\item \emph{General theory} (Section~\ref{sec:general}). The syntax separates
points from lengths, and the semantics makes explicit the window mode, the
rule that selects a window within a ball. Displacement commutes with the
bounded modalities if and only if the window mode is covariant, meaning that
translation of a point induces the same translation of its windows (Proposition~\ref{prop:disp}(iv)). For
open balls this holds exactly over groups (Lemma~\ref{lem:covariance}). Under
covariance a formula has an atomic-displacement normal form of at most
quadratic size, in which displacement applies only to atomic propositions and
their negations (Theorem~\ref{thm:adnf}). In linear time, windows compose by different laws in
discrete and in dense time (Proposition~\ref{prop:linear}).
\item \emph{The discrete logic: undecidability, expressiveness, equations}
(Sections~\ref{sec:eln}--\ref{sec:calculus}). The main theorem
(Theorem~\ref{thm:main}) is proved by a reduction from Post's Correspondence
Problem. A solution is encoded as a table, the ledger, whose rows are
concatenated in an $\omega$-word. The models of the reduction formula are
exactly the $\omega$-words that begin with the ledger of a solution, up to
one base letter (Proposition~\ref{prop:ledgers}). We determine the
degree of unsolvability, $\Sigma^0_1$ (Corollary~\ref{cor:prefix}), and establish
decidability with one quantified variable and without any one of the four
constructs (Propositions~\ref{prop:onevar} and~\ref{prop:necessary}). We also prove
undecidability of degree $\Pi^0_1$ with one universal quantifier and no
bounded modality (Proposition~\ref{prop:universal}), and the consequence for
validity (Corollary~\ref{cor:nonre}). Linear temporal logic embeds into
$\ELN$ (Proposition~\ref{prop:ltl}). The degree of unsolvability of the
full logic, $\Sigma^1_1$, and the first-order bounds are recalled
from~\cite{DZ26} (Theorems~\ref{thm:forward} and~\ref{thm:analytic}) and
complemented by a proper inclusion among the $\omega$-regular languages
(Theorem~\ref{thm:hierarchy}). The accepting runs of a B\"uchi automaton
are definable, so that an $\omega$-regular language is a projection of an
$\ELN$-definable one (Proposition~\ref{prop:buchi},
Corollary~\ref{cor:projection}). A sound equational calculus derives a
normal form without negation (Theorems~\ref{thm:sound}
and~\ref{thm:positive}).
\item \emph{Rational time: model checking on intervals}
(Section~\ref{sec:qel}). A labelled interval ensemble (LINE) presents an
interpretation over $\Rat$ by a finite union of rational intervals for each
atomic proposition (Definition~\ref{def:best}). We recast the framework of~\cite{qel} in the
general notation and treat point and length quantifiers explicitly. The
bounded modalities act on supports as dilation and erosion, computed by
arithmetic on endpoints (Proposition~\ref{prop:morph}). The satisfying
assignments of a formula form a rational semilinear set. This yields finite
evidence for quantification: a finite set of assignments sufficient to decide
the quantified formulas
(Theorem~\ref{thm:evidence}). For a fixed formula, model checking
takes time polynomial in the binary size of the LINE and of the point of
evaluation (Theorem~\ref{thm:datacomplexity}). Satisfiability, by contrast, is
$\Sigma^1_1$-complete for the rational logic. The lower bound follows from
interpreting the discrete logic in the linear-time logics
(Theorem~\ref{thm:lineartime}). The
fragment of the main theorem transfers with its two shared lengths
(Corollary~\ref{cor:threelengths}).
\item \emph{Events in real time: two window geometries} (Section~\ref{sec:mtl}).
A timed word is represented over $\Nat\times\Real$ by placing its $i$-th event
at position $i$ and at its timestamp, which keeps simultaneous events
distinct. Balls of
the metric $\lambda|i-j|+|r-s|$ bound position and time jointly and are not
covariant. Rectangular windows bound them independently and are covariant,
so the general laws apply. With rectangular windows, two shared lengths yield
undecidability already on finite timed words (Corollary~\ref{cor:timed}),
where metric temporal logic is decidable~\cite{worrell}. The resulting
questions appear in Section~\ref{sec:discussion}.
\end{enumerate}

\subsection{Related work}\label{subsec:pointers}

\emph{Temporal logics.} Linear temporal logic embeds into
$\ELN$, and the extension is proper. Wolper's property that a
letter holds at every even position, which is not expressible in linear
temporal logic~\cite{wolper1983}, is defined by $\forall x\,a_{2x}$. The
extension differs in kind from adding automata connectives or quantification
over sets of positions, both of which give the $\omega$-regular
languages~\cite{wolper1983,buchi62,Tho90}. The $\ELN$-definable languages
are incomparable with the $\omega$-regular languages~\cite{DZ26}. Among the
metric logics mentioned above, parametric temporal logic lies
at a decidability threshold~\cite{aetp2001}. What
distinguishes $\EL$ is that the same quantified variable may specify both a
window length and a displacement. Interval temporal logics take intervals as
primitive~\cite{hs1991}, whereas $\EL$ reaches intervals through windows. The
quantifiers also differ from those of first-order temporal
logic~\cite{fltl}: those range over individuals in a state, whereas the
quantifiers of $\EL$ range over points and lengths of the metric monoid.

\emph{Arithmetic.} The standard translation (Theorem~\ref{thm:forward}) maps
$\ELN$ into first-order logic with order and addition over $\omega$-words.
The terms of $\ELN$ become terms of that logic. Satisfiability in $\ELN$ is
nevertheless much harder than satisfiability in Presburger arithmetic,
because it asks whether unary predicates interpreting the
letters exist. This is an existential second-order question.
Its high degree of unsolvability is established in~\cite{DZ26}. For
Presburger arithmetic with uninterpreted unary predicates, the target of the
translation, validity is $\Pi^1_1$-complete~\cite{halpern1991}, and
satisfiability is undecidable already for universal sentences~\cite{HVW17}.
By comparison, the formulas of our main theorem use windows bounded by two
existentially quantified lengths and no universal quantifier, and those of
Proposition~\ref{prop:universal} use a single variable and no comparison of
terms.

\emph{Earlier and companion work.} This paper is the full version of the
preprint~\cite{zhang2024temporal}, which introduced linear-time logics of the
family under the name \emph{Temporal Ensemble Logic} and sketched a reduction
from Post's Correspondence Problem that places the top and the bottom
decomposition of the common word in consecutive parts of an $\omega$-word.
The proof in Section~\ref{sec:pcp} is complete and rests on a different
encoding, in which each row of a table records the application of one domino
to both decompositions. This encoding strengthens the result to existential
formulas. Joint work with Droste~\cite{DZ26} determines the degree of
unsolvability of $\ELN$, separates the $\omega$-regular languages from the
$\ELN$-definable ones, identifies an existential fragment with NP-complete
satisfiability, and treats model checking on finite words. The normal form
of Theorem~\ref{thm:adnf} stems from this collaboration, as do the setting
of Section~\ref{sec:mtl} and questions of Section~\ref{sec:discussion}. For
the rational logic see~\cite{qel}.

\subsection{Organization}
Section~\ref{sec:general} defines the family and proves its general laws.
Section~\ref{sec:eln} specializes to $\ELN$ and embeds linear temporal logic.
Section~\ref{sec:pcp} proves the main theorem and its sharpness.
Sections~\ref{sec:hierarchy} and~\ref{sec:calculus} treat expressiveness and
the equational calculus. Section~\ref{sec:qel} treats model checking and
satisfiability over rational time; Section~\ref{sec:mtl} treats timed words.
Section~\ref{sec:discussion} collects open problems, and
Section~\ref{sec:conclusion} concludes. Proofs appear in the body of the
paper. The main theorem, its prerequisites and the results establishing its
sharpness are proved in full. Numbered results from companion papers
are cited and accompanied by proofs or proof outlines in the present
two-sorted notation. For Theorem~\ref{thm:analytic}, we specify the source of
the lower bound. Routine structural inductions are stated briefly.

\section{Ensemble Logic over a Metric Monoid}\label{sec:general}

A logic of the family rests on a monoid of points, a metric, and a
rule selecting a window within each ball. Points serve as both locations of
evaluation and displacements. A commutative monoid accommodates these roles;
translation invariance links them by making the distance moved independent
of the starting point. We do not assume inverses, since time may have a
first instant. The laws below are stated under the assumptions they require.
Covariance yields the correspondence between positions and suffixes of
$\omega$-words in Section~\ref{sec:eln}. This correspondence supports the
reduction of Section~\ref{sec:pcp}. The laws of displacement and additive
decomposition become axioms of the calculus in Section~\ref{sec:calculus}.
Our presentation follows~\cite{zhang2024temporal,qel}, with separate sorts for
points and lengths.

\subsection{Metric monoids}\label{subsec:mm}

\begin{defi}[Metric monoid]\label{def:mm}
A \emph{metric monoid} is a structure $\mathcal{M}=(M,+,0,d)$ in which
$(M,+,0)$ is a commutative monoid and $d\colon M\times M\to\Realp$ is a metric
on $M$ that is \emph{translation invariant}:
\[
d(a+c,\,b+c)=d(a,b)\qquad\text{for all }a,b,c\in M .
\]
The elements of $M$ are called \emph{points}. The \emph{norm} of a point $a$
is $\norm{a}:=d(a,0)$, and the \emph{open ball} of radius $\ell>0$ about $a$
is $B(a,\ell):=\{b\in M\mid d(a,b)<\ell\}$. A \emph{length domain} for
$\mathcal{M}$ is a nonempty set $\Lambda\subseteq\Real_{>0}$ that is closed
under addition and contains the norm of every point other than $0$; its
elements are called \emph{lengths}.
\end{defi}

\noindent
A metric monoid comes with a length domain; we write $M$ for $\mathcal{M}$.
For a set $S$ of points and a point $c$ we put
$S+c:=\{b+c\mid b\in S\}$. Lengths bound windows. The conditions on
$\Lambda$ ensure that a sum of lengths and norms of points is a length
whenever one summand is a length.

\begin{exa}[Metric monoids]\label{ex:carriers}
The following four classes are used in this paper; Figure~\ref{fig:carriers}
shows an open ball in each. Addition is the usual one, taken componentwise
on products.
\begin{enumerate}[label=\textup{(\roman*)},leftmargin=2.2em]
\item $\Nat$ and $\Int$ with $d(a,b)=|a-b|$ and length domain $\Natp$.
\item $\Ratp$ and $\Rat$ with $d(a,b)=|a-b|$ and length domain $\Rat_{>0}$;
$\Realp$ and $\Real$ with the same metric and length domain $\Real_{>0}$.
\item $\Real^k$ with the Euclidean metric and length domain $\Real_{>0}$;
$\Int^k$ with the $\ell_1$ metric $d(a,b)=\sum_{j}|a_j-b_j|$ and length domain
$\Natp$.
\item The product $M_1\times M_2$ of two metric monoids, with the maximum or
with a positively weighted sum of the two metrics, and length domain
$\Real_{>0}$. Section~\ref{sec:mtl} uses $\Nat\times\Real$ with
$d_\lambda((i,r),(j,s))=\lambda|i-j|+|r-s|$ for a fixed $\lambda>0$.
\end{enumerate}
\end{exa}

\begin{figure}[htbp]
\centering
\begin{tikzpicture}[font=\small]
\begin{scope}[shift={(0,0)}]
 \draw[gray!60] (0,0)--(5.1,0);
 \foreach \i in {0,...,7}{\fill[gray!70] (0.7*\i,0) circle (1.5pt);}
 \foreach \i in {2,3,4}{\fill[inkblue] (0.7*\i,0) circle (2.3pt);}
 \draw[inkred,line width=0.9pt] (2.1,0) circle (4pt);
 \node[below] at (0,-0.15){\scriptsize $0$};
 \node[below] at (2.1,-0.15){\scriptsize $a$};
 \node[above,inkblue] at (2.1,0.2){\scriptsize $B(a,2)=\{a-1,a,a+1\}$};
 \node at (2.5,-0.95){(a) $\Nat$ with $|a-b|$};
\end{scope}
\begin{scope}[shift={(7.2,0)}]
 \draw[->,gray!70,line width=0.8pt] (0,0)--(6.3,0);
 \draw (0,0.08)--(0,-0.08); \node[below] at (0,-0.15){\scriptsize $0$};
 \draw[inkblue,line width=2pt] (1.6,0)--(4.0,0);
 \draw[inkblue,fill=white] (1.6,0) circle (2pt) (4.0,0) circle (2pt);
 \fill[inkred] (2.8,0) circle (1.8pt);
 \node[below] at (2.8,-0.15){\scriptsize $a$};
 \node[above,inkblue] at (2.8,0.2){\scriptsize $B(a,\ell)=(a-\ell,\,a+\ell)\cap M$};
 \node at (3.1,-0.95){(b) $\Ratp$ or $\Realp$ with $|a-b|$};
\end{scope}
\begin{scope}[shift={(1.6,-3.35)},x=0.36cm,y=0.36cm]
 \foreach \x in {-3,...,3}{\foreach \y in {-3,...,3}{\fill[gray!60] (\x,\y) circle (1.1pt);}}
 \foreach \x/\y in {0/0,1/0,-1/0,0/1,0/-1,2/0,-2/0,0/2,0/-2,1/1,1/-1,-1/1,-1/-1}{%
   \fill[inkblue] (\x,\y) circle (2pt);}
 \draw[dashed,inkblue] (3,0)--(0,3)--(-3,0)--(0,-3)--cycle;
 \draw[inkred,line width=0.9pt] (0,0) circle (3.4pt);
 \node at (0,-4.7){(c) $\Int^2$ with the $\ell_1$ metric};
\end{scope}
\begin{scope}[shift={(6.7,-3.35)}]
 \fill[inkblue!18] (0,0) circle (1.08);
 \draw[dashed,inkblue] (0,0) circle (1.08);
 \draw[->,inkblue] (0,0)--(0.76,0.76) node[midway,above left=-2pt]{\scriptsize $\ell$};
 \fill[inkred] (0,0) circle (1.7pt);
 \node at (0,-1.7){(d) $\Real^2$, Euclidean};
\end{scope}
\begin{scope}[shift={(10.1,-3.35)},x=0.56cm,y=0.36cm]
 \foreach \k in {0,...,6}{\draw[gray!55] (\k,-3.1)--(\k,3.1);}
 \draw[dashed,inkblue] (0.5,0)--(3,2.5)--(5.5,0)--(3,-2.5)--cycle;
 \draw[inkblue,line width=2pt] (3,-2.5)--(3,2.5) (2,-1.5)--(2,1.5) (4,-1.5)--(4,1.5)
      (1,-0.5)--(1,0.5) (5,-0.5)--(5,0.5);
 \fill[inkred] (3,0) circle (1.7pt);
 \node at (3,-4.7){(e) $\Nat\times\Real$ with $d_\lambda$};
\end{scope}
\end{tikzpicture}
\caption{Open balls (blue) centred at the marked points (red) in five metric
monoids: discrete time, dense time, the $\ell_1$ grid, the Euclidean plane,
and $\Nat\times\Real$ with $d_\lambda$. The radius is $3$ in (c) and $2.5$
in (e), where $\lambda=1$. Dashed boundaries are excluded from the balls.}
\label{fig:carriers}
\end{figure}

Translation invariance and the triangle inequality give the usual norm
properties without requiring subtraction in the monoid.

\begin{lem}[Elementary properties]\label{lem:mm}
Let $M$ be a metric monoid and let $a,b,u,v$ be points.
\begin{enumerate}[label=\textup{(\roman*)},leftmargin=2.2em]
\item $M$ is cancellative: $a+u=b+u$ implies $a=b$.
\item $d(a+u,a+v)=d(u,v)$; in particular $d(a,a+u)=\norm{u}$.
\item $\norm{a+b}\leq\norm{a}+\norm{b}$, and $\norm{a}=0$ if and only if $a=0$.
\item $\bigl|\norm{a}-\norm{b}\bigr|\leq d(a,b)$.
\end{enumerate}
\end{lem}
\begin{proof}
(i) If $a+u=b+u$ then $d(a,b)=d(a+u,b+u)=0$. (ii) By commutativity and
translation invariance, $d(a+u,a+v)=d(u+a,v+a)=d(u,v)$; in particular
$d(a,a+u)=d(0,u)=\norm{u}$. (iii) By the triangle inequality and translation
invariance,
$\norm{a+b}\leq d(a+b,b)+d(b,0)=\norm{a}+\norm{b}$; and $d(a,0)=0$ exactly
when $a=0$. (iv) is the triangle inequality for the points $a$, $b$ and $0$.
\end{proof}

\noindent
When $(M,+,0)$ is a group, the norm determines the metric:
$d(a,b)=\norm{a-b}$. In a monoid such as $\Nat$, this identity applies when
the difference exists in the monoid. The following metric monoids represent
linear time and include those of Example~\ref{ex:carriers}(i) and~(ii).

\begin{defi}[Linear metric monoids]\label{def:linear}
Let $G\neq\{0\}$ be a subgroup of $(\Real,+)$, and let $G_{\geq0}:=G\cap\Realp$.
With the metric $d(a,b)=|a-b|$ and the length domain
$\Lambda:=G\cap\Real_{>0}$, both $G$ and $G_{\geq0}$ are metric monoids. We call
$G$ a \emph{signed} and $G_{\geq0}$ a \emph{forward linear metric monoid}; a
\emph{linear metric monoid} is one of the two. It is \emph{discrete with unit
$\beta>0$} if $G=\beta\Int$, and \emph{dense} if $G$ is dense in $\Real$.
\end{defi}

\noindent
Thus $\Nat$, $\Ratp$ and $\Realp$ are forward linear; $\Int$, $\Rat$ and $\Real$
are signed linear. The monoids $\Nat$ and $\Int$ are discrete with unit $1$.
In a linear metric monoid a length is also a point, $\Lambda\subseteq M$.
The lengths are precisely the differences $b-a$ of points $a<b$.

\begin{lem}[Dichotomy]\label{lem:dichotomy}
Every linear metric monoid is either discrete or dense, and the two cases
exclude each other.
\end{lem}
\begin{proof}
Put $\beta:=\inf(G\cap\Real_{>0})$. If $\beta=0$, then $G$ contains
arbitrarily small positive elements, whose multiples meet every nonempty open
interval. If $\beta>0$, then $\beta\in G$:
otherwise $G$ would contain two distinct elements of the interval
$(\beta,2\beta)$, whose difference is a positive element of $G$ below
$\beta$. Division with remainder then gives $G=\beta\Int$, which is not
dense.
\end{proof}

\noindent
The proof relies on subtraction in $G$. The dichotomy therefore does not
extend to arbitrary additive submonoids of $\Real$. The two cases have
different laws of window composition (Section~\ref{subsec:linear}).

\subsection{Windows}\label{subsec:windows}
A bounded modality need not inspect the whole ball. In time, one usually
wants only the points that lie ahead. A window mode selects the points of a
ball that belong to the window.

\begin{defi}[Windows]\label{def:windows}
A \emph{window mode} for a metric monoid $M$ is a map
$W\colon M\times\Lambda\to2^{M}$ such that $a\in W(a,\ell)\subseteq B(a,\ell)$,
and $W(a,\ell)\subseteq W(a,\ell')$ whenever $\ell\leq\ell'$. The set
$W(a,\ell)$ is the \emph{window of length $\ell$ at $a$}. The \emph{full mode}
is $W=B$. For a linear metric monoid, the \emph{future mode} is
$W^{+}(a,\ell):=[a,a+\ell)\cap M$ and the \emph{past mode} is
$W^{-}(a,\ell):=(a-\ell,a]\cap M$. A window mode is \emph{covariant} if
$W(a+c,\ell)=W(a,\ell)+c$ for all points $a,c$ and all lengths $\ell$.
\end{defi}

The full mode suits spatial reasoning, where no direction is distinguished.
The future mode suits temporal reasoning: $[a,a+\ell)$ is the interval of
duration $\ell$ beginning at $a$. Covariance requires translation to carry
windows along with their points (Figure~\ref{fig:modes}). For these three
modes, $W(a,\ell)+c\subseteq W(a+c,\ell)$. Equality can fail at a boundary
of the monoid; the following lemma determines when it holds.

\begin{lem}[Covariance]\label{lem:covariance}
Let $M$ be a metric monoid.
\begin{enumerate}[label=\textup{(\roman*)},leftmargin=2.2em]
\item The full mode is covariant if and only if $(M,+,0)$ is a group.
\item On every linear metric monoid the future mode is covariant; on every
signed linear metric monoid so is the past mode.
\item On a forward linear metric monoid neither the full mode nor the past
mode is covariant.
\end{enumerate}
\end{lem}
\begin{proof}
For a point $b$ of $W(a+c,\ell)$, the question is whether it has the form
$b'+c$ for a point $b'$ of $W(a,\ell)$. (i) In a group take $b':=b-c$;
then $d(a,b')=d(a+c,b)<\ell$. Conversely, suppose that the full mode is
covariant and let $c\in M$. Since $\Lambda$ is nonempty and closed under
addition, it contains a length $\ell>\norm{c}$. Then
$0\in B(c,\ell)=B(0,\ell)+c$, so that $b+c=0$ for some point $b$. Every point
thus has an additive inverse, and $M$ is a group. (ii) For the future mode,
$a+c\leq b<a+c+\ell$ gives $b-c\geq a$, which is nonnegative in the forward
case; hence $b':=b-c$ is a point of $[a,a+\ell)$. The same computation covers
the past mode on a signed linear metric monoid. (iii) Choose a length $\ell$,
with $\ell>\beta$ in the discrete case, and put $c:=\ell$. Then
$W(0,\ell)+c\subseteq[c,\infty)$ for both modes, whereas $B(c,\ell)$ and
$W^{-}(c,\ell)$ contain a point below $c$: the point $c-\beta$ in the
discrete case, and any point of $M$ in $(0,c)$ in the dense case.
\end{proof}

\begin{figure}[htbp]
\centering
\begin{tikzpicture}[font=\small,>={Latex[length=2mm,width=1.6mm]}]
\foreach \y/\lab in {0/full mode,-0.95/future mode,-1.9/past mode}{
  \draw[gray!60,->] (0,\y)--(5.6,\y);
  \node[left] at (-0.05,\y){\scriptsize \lab};}
\draw[inkblue,line width=2pt] (1.2,0)--(4.4,0);
\draw[inkblue,fill=white] (1.2,0) circle (2pt) (4.4,0) circle (2pt);
\fill[inkred] (2.8,0) circle (1.8pt);
\draw[inkblue,line width=2pt] (2.8,-0.95)--(4.4,-0.95);
\draw[inkblue,fill=white] (4.4,-0.95) circle (2pt);
\fill[inkblue] (2.8,-0.95) circle (2.4pt); \fill[inkred] (2.8,-0.95) circle (1.4pt);
\draw[inkblue,line width=2pt] (1.2,-1.9)--(2.8,-1.9);
\draw[inkblue,fill=white] (1.2,-1.9) circle (2pt);
\fill[inkblue] (2.8,-1.9) circle (2.4pt); \fill[inkred] (2.8,-1.9) circle (1.4pt);
\foreach \x/\lab in {1.2/{a-\ell},2.8/{a},4.4/{a+\ell}}{
  \draw[gray!50,densely dotted] (\x,0.25)--(\x,-2.15);
  \node[below] at (\x,-2.2){\scriptsize $\lab$};}
\begin{scope}[shift={(7.4,0.1)},x=0.58cm]
 \node[anchor=west] at (-0.4,1.25){\scriptsize future mode on $\Nat$: $W^{+}(0,3)+5=W^{+}(5,3)$};
 \foreach \i in {0,...,8}{\fill[gray!60] (\i,0) circle (1.4pt);
    \node[below] at (\i,-0.12){\scriptsize \i};}
 \foreach \i in {0,1,2,5,6,7}{\fill[inkblue] (\i,0) circle (2.3pt);}
 \draw[->,inkred] (1,0.2) to[bend left=24] node[above=-1pt]{\scriptsize $+5$} (6,0.2);
 \node[anchor=west] at (-0.4,-1.05){\scriptsize full mode on $\Nat$: $B(0,3)+5\neq B(5,3)$};
 \foreach \i in {0,...,8}{\fill[gray!60] (\i,-1.7) circle (1.4pt);
    \node[below] at (\i,-1.82){\scriptsize \i};}
 \foreach \i in {0,1,2,5,6,7}{\fill[inkblue] (\i,-1.7) circle (2.3pt);}
 \foreach \i in {3,4,5,6,7}{\draw[inkred,line width=0.8pt] (\i,-1.7) circle (3.6pt);}
\end{scope}
\end{tikzpicture}
\caption{Left: windows of length $\ell$ at a point $a$ in a dense signed
linear metric monoid, for the full, future and past modes. Hollow endpoints
are excluded. Right: covariance and its failure over $\Nat$
(Lemma~\ref{lem:covariance}). Translating the future window at $0$ by $5$
gives the future window at $5$. In contrast, translating
$B(0,3)=\{0,1,2\}$ by $5$ (blue) gives a proper subset of
$B(5,3)=\{3,\ldots,7\}$ (red rings), because $\Nat$ has no points below $0$.}
\label{fig:modes}
\end{figure}

\subsection{Syntax}\label{subsec:syntax}
Fix a metric monoid $M$ with length domain $\Lambda$ and a window mode $W$.
The two sorts of terms distinguish displacements, which are points, from
window bounds, which are lengths. Point terms are built by addition. Length
terms also allow norms of point terms, so a window bound may depend on a
displacement.

\begin{defi}[Affine terms]\label{def:terms}
Fix disjoint countable sets of \emph{point variables} $x,y$ and of
\emph{length variables} $z$. \emph{Affine terms} $u,v$ and
\emph{positive affine terms} $s,t$ are given by
\[
u,v\;::=\;c\mid x\mid u+v
\qquad\text{and}\qquad
s,t\;::=\;r\mid z\mid s+t\mid s+\norm{u},
\]
where $c$ ranges over a set of constants for points, which includes $0$, and
$r$ over a set of constants for lengths; the two sets belong to the
specification of a logic of the family. Each constant names a point or a
length, and we use the same symbol for a constant and for what it names. An
\emph{assignment} $\eta$ maps every point variable to a point and
every length variable to a length. It extends to terms by $\eta(c)=c$,
$\eta(u+v)=\eta(u)+\eta(v)$, $\eta(r)=r$, $\eta(s+t)=\eta(s)+\eta(t)$ and
$\eta(s+\norm{u})=\eta(s)+\norm{\eta(u)}$. Under $\eta$, an affine term $u$
\emph{denotes} the point $\eta(u)$, and a positive affine term $t$ denotes the
length $\eta(t)$.
\end{defi}

\noindent
Associativity and commutativity allow us to collect the summands of a term.
An affine term then has the \emph{affine form}
\[
u\;=\;c+k_1x_1+\cdots+k_mx_m ,
\]
with a point $c$, distinct point variables $x_i$ and positive integers $k_i$,
where $kx$ denotes the $k$-fold sum $x+\cdots+x$. Terms with the same affine
form denote the same point under a given assignment. A positive affine term
similarly collects into a sum of a length, multiples of length variables,
and norms of affine terms. The length summand may be absent only when a
length variable is present. The conditions on $\Lambda$ in
Definition~\ref{def:mm} ensure that such a term denotes a length under any
assignment, hence the name positive affine. Positivity matters because a
window contains its point of evaluation (Definition~\ref{def:windows}),
whereas a ball of radius $0$ is empty. The norm-based bounds of~\cite{qel}
require a separate positivity proviso; here the grammar ensures it.

\begin{conv}[Lengths as points]\label{conv:coercion}
In a linear metric monoid, where $\Lambda\subseteq M$, we admit length
constants and length variables in affine terms,
$u::=c\mid r\mid x\mid z\mid u+v$. A length thus also serves to displace a formula
into the future. In a forward linear metric monoid, moreover, $\norm{u}=u$,
and we drop the norm bars: a positive affine term is then an affine term with
at least one summand that is a length constant or a length variable. In a
signed linear metric monoid $\norm{u}=|u|$, and the norm bars are retained.
\end{conv}

\begin{defi}[Formulas]\label{def:formulas}
Let $\Prop$ be a nonempty set of \emph{atomic propositions}. The
\emph{formulas} of $\EL$ over $M$ are given by
\[
\varphi,\psi\;::=\;p\mid\neg\varphi\mid\varphi\wedge\psi\mid\varphi\vee\psi\mid
\varphi_u\mid\mdiamond_t\varphi\mid\boldBox_t\varphi\mid
\exists x\,\varphi\mid\forall x\,\varphi\mid\exists z\,\varphi\mid\forall z\,\varphi,
\]
where $p\in\Prop$, $u$ is an affine term, $t$ is a positive affine term, $x$ is
a point variable and $z$ is a length variable. The three \emph{modal
operators} are \emph{displacement} $\varphi_u$, \emph{bounded existence}
$\mdiamond_t\varphi$ and \emph{bounded universality} $\boldBox_t\varphi$; the
last two are the \emph{bounded modalities}. The term in the subscript of a
modal operator is its \emph{parameter}. We use the abbreviations
$\top:=p\vee\neg p$ for a fixed $p\in\Prop$, $\bot:=\neg\top$ and
$\varphi\rightarrow\psi:=\neg\varphi\vee\psi$.
\end{defi}

\noindent
Variables occur in parameters only. A quantifier binds the occurrences of its
variable within its scope; modal operators do not bind variables. Free and bound
occurrences, \emph{closed} formulas, and \emph{quantifier-free} formulas are
defined as usual. A variable is \emph{shared} in a formula if it occurs in
more than one parameter, as the length variable $n$ does in the formula of the
Introduction; a shared length variable is called a \emph{shared length}. We
write $\varphi[\xi:=v]$ for the substitution of a term $v$
for the free occurrences of a variable $\xi$ of the same sort, with bound
variables renamed so that no variable of $v$ is captured. A variable is
\emph{fresh} for a formula or a term in which it does not occur. The
\emph{size} $|\varphi|$ of a formula is the number of nodes of its syntax
tree, those of its terms included. Subscripts that are labels or list
indices, as in $\varphi_{\mathcal{A}}$ and $a_1\cdots a_r$, are parts of names, rather than displacements.

\subsection{Semantics}\label{subsec:semantics}
An \emph{interpretation} is a map $\alpha\colon M\to2^{\Prop}$; it specifies
which atomic propositions hold at each point.

\begin{defi}[Satisfaction]\label{def:sat}
Let $\alpha$ be an interpretation, $\eta$ an assignment, and $a$ a point, the
\emph{point of evaluation}. The relation $(\alpha,\eta,a)\models\varphi$, read
\emph{$\varphi$ holds at $a$ under $\alpha$ and $\eta$}, is defined by the
following clauses, in which $W$ is the fixed window mode:
\[
\begin{array}{lcl}
(\alpha,\eta,a)\models p &\text{iff}& p\in\alpha(a);\\
(\alpha,\eta,a)\models \neg\varphi &\text{iff}& (\alpha,\eta,a)\not\models\varphi;\\
(\alpha,\eta,a)\models \varphi\wedge\psi &\text{iff}&
   (\alpha,\eta,a)\models\varphi \text{ and } (\alpha,\eta,a)\models\psi;\\
(\alpha,\eta,a)\models \varphi\vee\psi &\text{iff}&
   (\alpha,\eta,a)\models\varphi \text{ or } (\alpha,\eta,a)\models\psi;\\
(\alpha,\eta,a)\models \varphi_u &\text{iff}& (\alpha,\eta,\,a+\eta(u))\models\varphi;\\
(\alpha,\eta,a)\models \mdiamond_t\varphi &\text{iff}&
   (\alpha,\eta,b)\models\varphi \text{ for some point } b\in W(a,\eta(t));\\
(\alpha,\eta,a)\models \boldBox_t\varphi &\text{iff}&
   (\alpha,\eta,b)\models\varphi \text{ for every point } b\in W(a,\eta(t));\\
(\alpha,\eta,a)\models \exists x\,\varphi &\text{iff}&
   (\alpha,\eta[x\mapsto b],a)\models\varphi \text{ for some point } b;\\
(\alpha,\eta,a)\models \exists z\,\varphi &\text{iff}&
   (\alpha,\eta[z\mapsto\ell],a)\models\varphi \text{ for some length } \ell;
\end{array}
\]
and dually for $\forall x$ and $\forall z$. Here $\eta[\xi\mapsto e]$ agrees
with $\eta$ except that it maps $\xi$ to $e$.
\end{defi}

\noindent
Combining modal operators lets us evaluate a formula throughout a window
at a displaced point (Figure~\ref{fig:operators-plane}).

\begin{figure}[htbp]
\centering
\begin{tikzpicture}[font=\small,>={Latex[length=2mm,width=1.6mm]}]
\fill[inkblue!15] (0,0) circle (1.35);
\draw[dashed,inkblue] (0,0) circle (1.35);
\fill[inkred!12] (4.6,1.0) circle (0.85);
\draw[dashed,inkred] (4.6,1.0) circle (0.85);
\draw[->,inkblue] (0,0)--(-0.95,0.95) node[midway,above right=-2pt]{\scriptsize $t$};
\draw[->,inkred] (4.6,1.0)--(5.2,1.6) node[midway,above left=-3pt]{\scriptsize $s$};
\draw[->,inkred,line width=1pt] (0,0)--(4.6,1.0) node[midway,below=2pt]{$u$};
\fill (0,0) circle (1.8pt) node[below left=1pt]{$a$};
\fill[inkred] (4.6,1.0) circle (1.8pt);
\node[inkred,anchor=west] at (5.5,0.5){$a+u$};
\fill[inkblue] (0.55,-0.7) circle (1.5pt) node[right=1pt]{\scriptsize $b$};
\node[inkblue,align=left,anchor=west] at (-2.9,-2.05)
  {$\mdiamond_t\varphi$ at $a$: $\varphi$ at \emph{some} $b$ with $d(a,b)<t$\\[1pt]
   $\boldBox_t\varphi$ at $a$: $\varphi$ at \emph{every} $b$ with $d(a,b)<t$};
\node[inkred,align=left,anchor=west] at (4.6,-2.05)
  {$\varphi_u$ at $a$: $\varphi$ at the point $a+u$\\[1pt]
   $(\boldBox_s\varphi)_u$ at $a$: $\varphi$ throughout $B(a+u,s)$};
\end{tikzpicture}
\caption{Displacement, bounded existence and bounded universality in the
Euclidean plane with full windows. Each bounded modality is evaluated over
an open ball centred at its point of evaluation; displacement changes that
point before the bounded modality is evaluated.}
\label{fig:operators-plane}
\end{figure}

\begin{defi}[Support, equivalence, satisfiability]\label{def:entail}
The \emph{support} of a formula $\varphi$ under $\alpha$ and $\eta$ is the set
of points at which it holds,
\[
\sem{\varphi}_{\alpha,\eta}:=\{a\in M\mid(\alpha,\eta,a)\models\varphi\}.
\]
A formula $\varphi$ \emph{entails} $\psi$, written $\varphi\models\psi$, if
$\sem{\varphi}_{\alpha,\eta}\subseteq\sem{\psi}_{\alpha,\eta}$ for all $\alpha$
and $\eta$; the two are \emph{equivalent}, written $\varphi\equiv\psi$, if each
entails the other. A formula is \emph{satisfiable} if at least one of its
supports is nonempty, and \emph{valid} if each of its supports is $M$. The
formulas over $M$ with this satisfaction relation constitute \emph{the logic
$\EL(M)$}; the length domain, the window mode, the constants and $\Prop$ are
understood from the context.
\end{defi}

\begin{lem}[Coincidence and substitution]\label{lem:substitution}
Let $\varphi$ be a formula, $\alpha$ an interpretation and $a$ a point.
\begin{enumerate}[label=\textup{(\roman*)},leftmargin=2.2em]
\item If the assignments $\eta$ and $\eta'$ agree on the free variables of
$\varphi$, then $(\alpha,\eta,a)\models\varphi$ if and only if
$(\alpha,\eta',a)\models\varphi$.
\item If $v$ is a term of the sort of the variable $\xi$ and
$\varphi[\xi:=v]$ is a formula, then
$(\alpha,\eta,a)\models\varphi[\xi:=v]$ if and only if
$(\alpha,\eta[\xi\mapsto\eta(v)],a)\models\varphi$.
\end{enumerate}
\end{lem}
\begin{proof}
Prove the corresponding statements for terms, then proceed by induction on
$\varphi$. In the quantifier case, renaming bound variables prevents capture
of the variables of $v$.
\end{proof}

\noindent
By Lemma~\ref{lem:substitution}(i), the support of a closed formula is
independent of the assignment. We then write $(\alpha,a)\models\varphi$ and
$\sem{\varphi}_\alpha$. The proviso in Lemma~\ref{lem:substitution}(ii) is
needed only under Convention~\ref{conv:coercion}: a length variable in an
affine term must be replaced by a term without norm bars.

\begin{rem}[Terms and witnesses]\label{rem:predicate-free}
The only predicates are the atomic propositions; neither $u=v$ nor $s<t$ is a
formula. The language offers two complementary ways to assert existence.
The modality $\mdiamond_t\varphi$ asserts the existence of a witness in a
window without naming it. The quantifier in $\exists x\,\varphi_x$ binds a
point variable usable in several parameters, as in
$\exists x\,(\varphi_x\wedge\psi_{x+x})$. There is no primitive comparison
to confine its witness to a window. This difference shapes the reduction of
Section~\ref{sec:pcp} and the comparison with first-order logic in
Section~\ref{sec:hierarchy}.
\end{rem}

\begin{defi}[Linear-time logics]\label{def:tel}
The \emph{linear-time logics} of the family are the logics $\EL(M)$ for a
linear metric monoid $M$ in future mode. They include the \emph{discrete
logic} $\EL(\Nat)$, treated in Sections~\ref{sec:eln}--\ref{sec:calculus},
and the \emph{dense logics} $\EL(\Ratp)$, $\EL(\Realp)$ and $\EL(\Rat)$, whose
constants are the rational points and lengths; $\EL(\Rat)$ is the
\emph{rational logic}, treated in Section~\ref{sec:qel}.
\end{defi}

\noindent
In a linear-time logic, $\boldBox_t\varphi$ reads ``$\varphi$ holds from now
until $t$ has elapsed'', and $\mdiamond_t\varphi$ reads ``$\varphi$ holds at
some time before $t$ has elapsed''. For forward linear $M$ these logics were
introduced in~\cite{zhang2024temporal} under the name Temporal Ensemble Logic.

\subsection{Laws}\label{subsec:laws}
For a covariant window mode, translating an interpretation has the same
effect as translating the point of evaluation.

\begin{prop}[Translation of interpretations]\label{prop:reference}
Let the window mode be covariant. For a point $c$ let $\alpha^{c}$ be the
interpretation $b\mapsto\alpha(b+c)$. Then for every formula $\varphi$,
assignment $\eta$ and point $a$,
\[
(\alpha^{c},\eta,a)\models\varphi
\quad\Longleftrightarrow\quad
(\alpha,\eta,a+c)\models\varphi .
\]
\end{prop}
\begin{proof}
By induction on $\varphi$. Displacement uses commutativity of addition.
For a bounded modality, covariance says that $b\mapsto b+c$ maps
$W(a,\eta(t))$ onto $W(a+c,\eta(t))$. The Boolean and quantifier cases are
immediate.
\end{proof}

\noindent
Displacement commutes with a quantifier when its variable is fresh for the
displacement term. It commutes with the bounded modalities precisely when
the window mode is covariant.

\begin{prop}[Laws of displacement]\label{prop:disp}
For all formulas $\varphi,\psi$, affine terms $u,v$ and positive affine terms
$t$:
\begin{enumerate}[label=\textup{(\roman*)},leftmargin=2.2em]
\item $\varphi_0\equiv\varphi$; $(\varphi_u)_v\equiv\varphi_{u+v}$; and
$\varphi_u\equiv\varphi_v$ whenever $u$ and $v$ denote the same point under
every assignment;
\item $(\neg\varphi)_u\equiv\neg(\varphi_u)$,
$(\varphi\wedge\psi)_u\equiv\varphi_u\wedge\psi_u$ and
$(\varphi\vee\psi)_u\equiv\varphi_u\vee\psi_u$;
\item $(Q\xi\,\varphi)_u\equiv Q\xi\,(\varphi_u)$ for $Q\in\{\exists,\forall\}$
and every variable $\xi$ that is fresh for $u$;
\item if the window mode is covariant, then
$(\boldBox_t\varphi)_u\equiv\boldBox_t(\varphi_u)$ and
$(\mdiamond_t\varphi)_u\equiv\mdiamond_t(\varphi_u)$. Conversely, if these
equivalences hold for all $\varphi$, $u$ and $t$, then the window mode is
covariant.
\end{enumerate}
\end{prop}
\begin{proof}
Parts (i)--(iii) follow directly from Definition~\ref{def:sat}. In
(iii), the point that $u$ denotes is unchanged when $\xi$ is reassigned. For
(iv), let $u$ denote $c$ and $t$ denote $\ell$. At a point $a$ the formula
$(\boldBox_t\varphi)_u$ inspects the window $W(a+c,\ell)$, and
$\boldBox_t(\varphi_u)$ inspects the set $W(a,\ell)+c$; under covariance the
two coincide. Conversely, suppose that $W(a+c,\ell)\neq W(a,\ell)+c$, and let
$b$ belong to exactly one of the two sets. Let $x$ be a point variable and $z$
a length variable, let $\eta(x)=c$ and $\eta(z)=\ell$, and let the atomic
proposition $p$ hold everywhere except at $b$. Then exactly one of
$(\boldBox_zp)_x$ and $\boldBox_z(p_x)$ holds at $a$.
\end{proof}

\begin{prop}[Laws of the bounded modalities~\cite{qel}]\label{prop:modal}
For all formulas $\varphi,\psi$, affine terms $u$ and positive affine terms
$s,t$:
\begin{enumerate}[label=\textup{(\roman*)},leftmargin=2.2em]
\item \textup{($\mdiamond$-$\vee$)}\;
$\mdiamond_s\varphi\equiv\varphi\vee\mdiamond_s\varphi$ and
$\mdiamond_s(\varphi\vee\psi)\equiv\mdiamond_s\varphi\vee\mdiamond_s\psi$;
\item \textup{($\boldBox$-$\wedge$)}\;
$\boldBox_t\varphi\equiv\varphi\wedge\boldBox_t\varphi$ and
$\boldBox_t(\varphi\wedge\psi)\equiv\boldBox_t\varphi\wedge\boldBox_t\psi$;
\item \textup{(Duality)}\;
$\neg\boldBox_t\varphi\equiv\mdiamond_t\neg\varphi$ and
$\neg\mdiamond_t\varphi\equiv\boldBox_t\neg\varphi$;
\item \textup{(Reflexivity and monotonicity)}\;
$\boldBox_s\varphi\models\varphi\models\mdiamond_s\varphi$;
$\mdiamond_s\varphi\models\mdiamond_{s+t}\varphi$ and
$\boldBox_{s+t}\varphi\models\boldBox_s\varphi$; and if $\varphi\models\psi$
then $\varphi_u\models\psi_u$, $\mdiamond_s\varphi\models\mdiamond_s\psi$ and
$\boldBox_s\varphi\models\boldBox_s\psi$;
\item \textup{(Reach)}\; in full mode, and in future mode over a forward
linear metric monoid: whenever $\norm{\eta(u)}<\eta(t)$,
$(\alpha,\eta,a)\models\varphi_u$ implies
$(\alpha,\eta,a)\models\mdiamond_t\varphi$, and
$(\alpha,\eta,a)\models\boldBox_t\varphi$ implies
$(\alpha,\eta,a)\models\varphi_u$. In particular
$\varphi_u\models\mdiamond_{t+\norm{u}}\varphi$ and
$\boldBox_{t+\norm{u}}\varphi\models\varphi_u$;
\item \textup{(Barcan laws)}\; if the variable $\xi$ is fresh for $t$, then
$\forall\xi\,\boldBox_t\varphi\equiv\boldBox_t\forall\xi\,\varphi$ and
$\exists\xi\,\mdiamond_t\varphi\equiv\mdiamond_t\exists\xi\,\varphi$, while
$\exists\xi\,\boldBox_t\varphi\models\boldBox_t\exists\xi\,\varphi$ and
$\mdiamond_t\forall\xi\,\varphi\models\forall\xi\,\mdiamond_t\varphi$.
\end{enumerate}
\end{prop}
\begin{proof}
The claims follow from Definition~\ref{def:sat}. Reflexivity holds because
a window contains its point of evaluation. Monotonicity in the parameter
follows from $W(a,\ell)\subseteq W(a,\ell')$ for $\ell\leq\ell'$. For reach, $d(a,a+\eta(u))=\norm{\eta(u)}$ by
Lemma~\ref{lem:mm}(ii), so that $a+\eta(u)$ lies in the window of length
$\eta(t)$ at $a$. In future mode this uses $\eta(u)\geq0$. The Barcan
equivalences follow by exchanging two universal or two existential
quantifiers, since the window is independent of $\xi$.
\end{proof}

De Morgan's laws, the dualities of Proposition~\ref{prop:modal}(iii) and of
the quantifiers, and $\neg(\varphi_u)\equiv(\neg\varphi)_u$ allow a formula to
be put in an equivalent \emph{negation normal form}. Negation then applies
only to atomic propositions. We use this fact freely. A \emph{literal} is an atomic
proposition or its negation, and a formula is \emph{existential} if it is
built from literals by $\wedge$, $\vee$, the modal operators and existential
quantifiers. When the window mode is covariant, displacement can likewise be restricted
to literals.

\begin{defi}[ADNF]\label{def:adnf}
A formula is in \emph{atomic-displacement normal form} (ADNF) if it is in
negation normal form and every displacement in it is applied to a literal.
\end{defi}

\begin{thm}[Atomic-displacement normal form]\label{thm:adnf}
Let the window mode be covariant. Every formula $\varphi$ of $\EL(M)$ is
equivalent to a formula $\varphi^{\circ}$ in ADNF, computable from $\varphi$,
such that, up to a renaming of bound variables,
\begin{enumerate}[label=\textup{(\roman*)},leftmargin=2.2em]
\item $\varphi^{\circ}$ has the same quantifiers and the same bounded
modalities as a negation normal form of $\varphi$, nested in the same way and
with the same parameters;
\item the parameter of every displacement in $\varphi^{\circ}$ is the sum of
the parameters of the displacements along one branch of $\varphi$.
\end{enumerate}
In particular the size of $\varphi^{\circ}$ is at most quadratic in the size of
$\varphi$.
\end{thm}
\begin{proof}
Put $\varphi$ in negation normal form and rename its bound variables apart:
distinct quantifiers bind distinct variables, and no bound variable occurs
free. A bound variable may still occur in a displacement parameter inside
its own scope. The argument requires freshness only for displacement terms
\emph{outside} that scope which are moved across the quantifier. Renaming
ensures this condition.

Traverse the syntax tree while carrying an affine term, the accumulated
displacement, initially $0$. At a displacement by $u$, add $u$ to the
accumulated term and continue into the operand. At a Boolean connective,
carry the accumulated term into each operand. At a quantifier or a bounded
modality, keep that constructor and continue into its operand with the same
accumulated term. At a literal, apply the accumulated displacement, omitting
a displacement by $0$. Each step is an instance of
Proposition~\ref{prop:disp}: the steps across bounded modalities use
covariance, and those across quantifiers use the freshness just established.
This structural recursion terminates and produces an equivalent formula
in ADNF.

The construction preserves the arrangement of Boolean connectives,
quantifiers and bounded modalities, proving (i). The term accumulated at a
literal sums the displacement parameters on the branch from the root to
that literal, proving (ii). There are $O(|\varphi|)$ literal occurrences,
with accumulated terms of size $O(|\varphi|)$. The remaining constructors
and their parameters have total size $O(|\varphi|)$, so
$|\varphi^{\circ}|=O(|\varphi|^2)$.
\end{proof}

\noindent
In ADNF, affine terms displace literals and positive affine terms bound
windows. This separation simplifies the standard translation
(Section~\ref{sec:hierarchy}). The form is derived equationally in
Section~\ref{sec:calculus} and suggests the automata considered in
Section~\ref{sec:discussion}.

\begin{rem}[Quantifiers and windows]\label{rem:prenex}
The two entailments in Proposition~\ref{prop:modal}(vi) are proper.
Quantifiers therefore cannot in general be extracted from the scope of a
bounded modality. The classical procedure for obtaining prenex normal form
stops there. For instance, in $\EL(\Nat)$ in future
mode let $p$ hold at the point $2$ only, and let
$\psi:=p_n\wedge\boldBox_n\neg p$ for a length variable $n$, which says that
$p$ first holds exactly $n$ ahead. Then $\boldBox_2\exists n\,\psi$ holds at
the point $0$, with the witnesses $n=2$ at the point $0$ and $n=1$ at the
point $1$. Yet $\exists n\,\boldBox_2\psi$ fails at the initial point of evaluation: a witness may depend
on the point of the window.
\end{rem}

\subsection{Laws of linear time}\label{subsec:linear}
Throughout this subsection $M$ is a linear metric monoid in future mode, under
Convention~\ref{conv:coercion}.
In a \emph{forward} linear metric monoid the points reached from a point $a$
by displacement are exactly the points at or after $a$,
$\{a+b\mid b\in M\}=[a,\infty)\cap M$. For a fresh point variable $x$ we
define the \emph{unbounded modalities}
\[
\boldBox\varphi:=\forall x\,\varphi_x
\qquad\text{and}\qquad
\mdiamond\varphi:=\exists x\,\varphi_x ,
\]
read ``always from now on'' and ``now or at some time in the future''. In a
linear metric monoid, \emph{until} is defined from one shared length:
for a length variable $z$ fresh for $\varphi$ and $\psi$,
\[
\varphi\U\psi:=\psi\vee\exists z\,(\psi_z\wedge\boldBox_z\varphi).
\]

\begin{lem}[Until]\label{lem:until}
In a linear metric monoid $M$ in future mode, the formula $\varphi\U\psi$
holds at a point $a$ if and only if $\psi$ holds at some point $b\geq a$ and
$\varphi$ throughout $[a,b)\cap M$.
\end{lem}
\begin{proof}
The first disjunct is the case $b=a$. For $b>a$ the difference $\ell:=b-a$ is
a length, and every length arises in this way. When $z$ denotes $\ell$, the
formula $\psi_z$ holds at $a$ exactly when $\psi$ holds at $b$, and
$\boldBox_z\varphi$ holds at $a$ exactly when $\varphi$ holds throughout
$[a,b)\cap M$. Since $z$ is fresh, the choice of $\ell$ affects neither
$\varphi$ nor $\psi$ (Lemma~\ref{lem:substitution}(i)).
\end{proof}

\noindent
Here the shared variable $z$ serves two roles: displacement reaches the
witness, and the bounded modality inspects the window before it. This
construction underlies the embedding of Section~\ref{subsec:ltl}.
Section~\ref{sec:discussion} discusses its analogue for timed words.

\begin{prop}[Laws of linear time]\label{prop:linear}
Let $M$ be a linear metric monoid in future mode, let $\varphi$ be a formula,
and let $s,t$ be positive affine terms, in (i) and (ii) without norm bars, so
that they are also affine terms.
\begin{enumerate}[label=\textup{(\roman*)},leftmargin=2.2em]
\item \textup{(Additive decomposition)}\;
$\boldBox_{s+t}\varphi\equiv\boldBox_s\varphi\wedge(\boldBox_t\varphi)_s$ and
$\mdiamond_{s+t}\varphi\equiv\mdiamond_s\varphi\vee(\mdiamond_t\varphi)_s$.
\item \textup{(Unwinding)}\; If $M$ is forward linear, then
$\boldBox\varphi\equiv\boldBox_t\varphi\wedge(\boldBox\varphi)_t$ and
$\mdiamond\varphi\equiv\mdiamond_t\varphi\vee(\mdiamond\varphi)_t$.
\item \textup{(Commutation)}\;
$\boldBox_s\boldBox_t\varphi\equiv\boldBox_t\boldBox_s\varphi$ and
$\mdiamond_s\mdiamond_t\varphi\equiv\mdiamond_t\mdiamond_s\varphi$.
\item \textup{(Composition, dense time)}\; If $M$ is dense, then
$\boldBox_s\boldBox_t\varphi\equiv\boldBox_{s+t}\varphi$ and
$\mdiamond_s\mdiamond_t\varphi\equiv\mdiamond_{s+t}\varphi$.
\item \textup{(Composition, discrete time)}\; If $M$ is discrete with unit
$\beta$, and $\beta$ is a constant, then
$\boldBox_{s+\beta}\boldBox_t\varphi\equiv\boldBox_{s+t}\varphi$ and
$\mdiamond_{s+\beta}\mdiamond_t\varphi\equiv\mdiamond_{s+t}\varphi$, and
moreover $\boldBox_\beta\varphi\equiv\varphi\equiv\mdiamond_\beta\varphi$.
\end{enumerate}
\end{prop}
\begin{proof}
Fix an interpretation, an assignment and a point $a$, and write $s,t$ also
for the lengths that these terms denote. Parts (i) and (ii) hold because
$[a,a+s+t)$ is the disjoint union of $[a,a+s)$ and $[a+s,a+s+t)$, and
$[a,\infty)$ that of $[a,a+t)$ and $[a+t,\infty)$. For the remaining parts,
the formula $\boldBox_s\boldBox_t\varphi$ inspects at $a$ the points of
\[
 \bigcup_{b\in[a,a+s)\cap M}[b,b+t)\cap M
 \;=\;\begin{cases}
 [a,a+s+t)\cap M,&M\text{ dense},\\
 [a,a+s+t-\beta)\cap M,&M\text{ discrete with unit }\beta.
 \end{cases}
\]
The left-hand side is included in the right-hand side because $b<a+s$, and
$b\leq a+s-\beta$ in the discrete case. Conversely, let $c$ be a point of
the right-hand side. If $c<a+s$, take $b:=c$. Otherwise take for $b$ a point
of $(\max(a,c-t),\,a+s)$ in the dense case, and $b:=a+s-\beta$ in the
discrete case. The former exists by density. Both choices give $b\leq c<b+t$. The right-hand side is symmetric in $s$ and $t$, which gives
(iii), since $M$ is dense or discrete (Lemma~\ref{lem:dichotomy}). It gives
(iv) directly, and (v) on replacing $s$ by $s+\beta$. The last claim of (v)
is $[a,a+\beta)\cap M=\{a\}$. The laws for $\mdiamond$ follow by duality.
\end{proof}

Additive decomposition splits a window into two consecutive windows.
Composition takes the union of the inner windows beginning at points of the
outer window. In discrete time the last inner window begins one unit before
the outer window ends. This accounts for $\beta$ in part (v); for $M=\Nat$
the unit is $1$.

\section{The Discrete Logic \texorpdfstring{$\ELN$}{EL(N)} and
\texorpdfstring{$\omega$}{omega}-Languages}\label{sec:eln}

From here through Section~\ref{sec:calculus}, we work over the metric monoid
$\Nat$ of Example~\ref{ex:carriers}(i), in future mode. This section prepares
the main theorem. Interpretations are identified with $\omega$-words, and
closed formulas define $\omega$-languages. The suffix lemma, an instance of
covariance, reduces satisfiability to truth at position $0$. The word
formulas of Section~\ref{subsec:words} spell out the dominoes in the
reduction of Section~\ref{sec:pcp}. An embedding of linear temporal logic
gives a first lower bound on expressiveness.

\subsection{Conventions}

\begin{conv}[The logic $\ELN$]\label{conv:two-sorts}
The points of $\Nat=\{0,1,2,\ldots\}$ are called \emph{positions}; the length
domain is $\Natp=\{1,2,\ldots\}$, so that the window of length $\ell$ at
position $i$ is $\{i,\ldots,i+\ell-1\}$, and the constants are the numerals,
those for lengths being positive.
Point variables are written $x,y,\ldots$; length variables are written
$n,m,N$. In accordance with
Convention~\ref{conv:coercion}, an affine term is a sum of numerals, point
variables and length variables. Collecting its summands gives a numeral plus
multiples of variables, and the term is a positive affine term if that numeral
is at least $1$ or a length variable is present. These are precisely the affine
terms that denote a positive number under every assignment. To see the
necessity of the condition, assign $0$ to every point variable and
$1$ to every length variable. In the text, positions are denoted by $i$, $h$, $j$ and $k$; in
Sections~\ref{sec:eln}--\ref{sec:calculus}, $a,b,c$ denote letters.
\end{conv}

\begin{lem}[Interdefinability of the sorts]\label{lem:sorts}
For every formula $\varphi$, length variable $n$ and point variable $x$,
\[
\exists n\,\varphi\;\equiv\;\exists x\,\varphi[n:=x+1]
\qquad\text{and}\qquad
\exists x\,\varphi\;\equiv\;\varphi[x:=0]\vee\exists n\,\varphi[x:=n],
\]
and dually for $\forall$, where the variable introduced on the right-hand side
is fresh for $\varphi$.
\end{lem}
\begin{proof}
A length is a position plus $1$, and a position is $0$ or a length; apply
Lemma~\ref{lem:substitution}(ii). These substitutions yield formulas: a
positive affine term remains positive when a length variable is replaced by
$x+1$, or a point variable by $0$ or by a length variable.
\end{proof}

\begin{rem}[Conventions in the literature]\label{rem:conventions}
The preprint~\cite{zhang2024temporal} lets all variables range over $\Natp$
and numbers positions from $1$; up to renumbering, this is the fragment of
$\ELN$ with length variables only. The companion paper~\cite{DZ26} lets
all variables range over $\Nat$ and takes as parameters of the bounded
modalities the affine terms whose numeral is at least $1$; this is exactly
the fragment with point variables only. By Lemma~\ref{lem:sorts}, either
one-sorted fragment is as expressive as $\ELN$. The first translation is
linear, and results proved under one convention transfer to the others. We
retain both sorts to distinguish their roles: length variables bound windows,
and point variables denote positions.
\end{rem}

\subsection{\texorpdfstring{$\omega$}{omega}-words}\label{subsec:words}
Let $\Sigma$ be a finite alphabet, whose elements are called \emph{letters}.
We write $\Sigma^{*}$ and $\Sigma^{+}$ for the sets of finite and of nonempty
finite words, $|w|$ for the length of a finite word $w$, $w^{\omega}$ for the
infinite repetition of a nonempty $w$, and $K^{\omega}$ for the set of infinite
concatenations of words from $K\subseteq\Sigma^{+}$.

\begin{defi}[$\omega$-words as interpretations]\label{def:words}
An \emph{$\omega$-word} over $\Sigma$ is a map $\sigma\colon\Nat\to\Sigma$,
written $\sigma=\sigma[0]\sigma[1]\sigma[2]\cdots$. In the \emph{alphabetic
presentation} of $\ELN$ the atomic propositions are the letters of $\Sigma$,
and the interpretations are the $\omega$-words: $\sigma$ is the interpretation
in which exactly the letter $\sigma[i]$ holds at position $i$,
\[
(\sigma,\eta,i)\models a\quad\text{iff}\quad\sigma[i]=a ,
\]
and satisfiability, validity, entailment and equivalence refer to
$\omega$-words only. The \emph{propositional presentation} is that of
Definition~\ref{def:sat}, with all interpretations
$\alpha\colon\Nat\to2^{\Prop}$. For $A\subseteq\Sigma$, the \emph{letter
class} $A$ is the formula $\bigvee_{a\in A}a$; in the alphabetic presentation
the class $\Sigma$ is equivalent to $\top$. An $\omega$-word $\sigma$ is a
\emph{model} of a closed formula $\varphi$ if $(\sigma,0)\models\varphi$. The
\emph{language} of $\varphi$ is the set of its models,
$\Lang(\varphi):=\{\sigma\in\Sigma^{\omega}\mid(\sigma,0)\models\varphi\}$,
and an $\omega$-language, a set of $\omega$-words, is \emph{$\ELN$-definable}
if it is the language of a closed formula. The $k$-th \emph{suffix} of an
interpretation $\alpha$, and in particular of an $\omega$-word, is the
interpretation $\alpha^{k}$ of Proposition~\ref{prop:reference}; to
distinguish it from a power of a word we write it $\alpha^{(k)}\colon
i\mapsto\alpha(k+i)$.
\end{defi}

\begin{lem}[Suffix invariance]\label{lem:suffix}
For every formula $\varphi$, interpretation $\alpha$, assignment $\eta$ and
positions $i,k$,
\[
(\alpha,\eta,i+k)\models\varphi
\quad\Longleftrightarrow\quad
(\alpha^{(k)},\eta,i)\models\varphi .
\]
In particular a formula is satisfiable if and only if it holds at position $0$
under some interpretation and assignment; and for a closed formula $\varphi$
of the alphabetic presentation, $(\sigma,i)\models\varphi$ if and only if
$\sigma^{(i)}\in\Lang(\varphi)$, so that $\varphi$ is satisfiable if and only
if $\Lang(\varphi)\neq\emptyset$.
\end{lem}
\begin{proof}
The future mode on $\Nat$ is covariant (Lemma~\ref{lem:covariance}(ii));
apply Proposition~\ref{prop:reference}.
\end{proof}

\begin{prop}[Presentations]\label{prop:presentations}
\begin{enumerate}[label=\textup{(\roman*)},leftmargin=2.2em]
\item Let $\Prop$ be finite and $\Sigma:=2^{\Prop}$. Replacing every atomic
proposition $p$ in a formula $\varphi$ by the letter class
$\{a\in\Sigma\mid p\in a\}$ yields a formula $\widehat{\varphi}$ of the
alphabetic presentation such that $(\alpha,\eta,i)\models\varphi$ if and only
if $(\alpha,\eta,i)\models\widehat{\varphi}$, where on the right $\alpha$ is
read as an $\omega$-word over $\Sigma$.
\item Conversely, let $\Prop:=\Sigma$ and
$\Uniq:=\forall x\,\bigl(\bigvee_{a\in\Sigma}(a\wedge\bigwedge_{b\neq a}\neg b)\bigr)_x$.
The interpretations under which $\Uniq$ holds at position $0$ are exactly the
$\omega$-words over $\Sigma$, and a formula $\varphi$ is satisfiable in the
alphabetic presentation if and only if $\varphi\wedge\Uniq$ is satisfiable in
the propositional presentation.
\end{enumerate}
\noindent
Consequently, for a fixed $\Prop$, respectively $\Sigma$, the satisfiability
problems of the two presentations reduce to each other in linear time.
\end{prop}
\begin{proof}
(i) The letter class that replaces $p$ holds at position $i$ of the
$\omega$-word $\alpha$ exactly when $p\in\alpha(i)$. Induction on
$\varphi$ gives the result. (ii) $\Uniq$ holds at position $0$ exactly when
each position carries exactly one letter. If $\varphi\wedge\Uniq$ holds at
position $i$ under $\alpha$, then by Lemma~\ref{lem:suffix} it holds at
position $0$ under $\alpha^{(i)}$, which is therefore an $\omega$-word. For a
fixed alphabet both translations are linear.
\end{proof}

For letter classes $A_1,\ldots,A_r$ the \emph{word formula} is
$\iota^{A_1\cdots A_r}:=A_1\wedge(A_2)_1\wedge\cdots\wedge(A_r)_{r-1}$, and
for a word $w=a_1\cdots a_r\in\Sigma^{+}$ we put
$\iota^{w}:=\iota^{\{a_1\}\cdots\{a_r\}}$. The displacement clause gives the
following characterization.

\begin{lem}[Word formulas]\label{lem:word}
$(\sigma,\eta,i)\models\iota^{A_1\cdots A_r}$ if and only if
$\sigma[i+g-1]\in A_g$ for $g=1,\ldots,r$. In particular, $\iota^{w}$ holds at
position $i$ exactly when $w$ occurs in $\sigma$ at position $i$.
\end{lem}

\subsection{Linear temporal logic}\label{subsec:ltl}
Linear temporal logic ($\LTL$) embeds directly into $\ELN$. Its structures
are already interpretations of $\ELN$. Displacement reaches a future
witness, and a window selects the positions before it. Thus next is
displacement by $1$, and until is the operator $\U$ of
Section~\ref{subsec:linear}. One shared length determines both the witness
and the preceding window.

Let $\Prop$ be finite. $\LTL$ is interpreted over the maps
$\alpha\colon\Nat\to2^{\Prop}$, the interpretations of
$\ELN$ in the propositional presentation. We use $\LTL$ with the next operator
$\X$ and the until operator $\U$. Until is \emph{strong}: a witness is
required. It is also \emph{reflexive}: the current position may be the
witness~\cite{pnueli1992temporal,tltextbook}:
\[
\theta::=p\mid\top\mid\neg\theta\mid\theta\wedge\theta
              \mid\X\theta\mid\theta\U\theta
\qquad(p\in\Prop).
\]
The Boolean clauses are classical, $\alpha,i\models_{\LTL}p$ if and only if
$p\in\alpha(i)$, and
\begin{align*}
\alpha,i\models_{\LTL}\X\theta
 &\text{ iff } \alpha,i+1\models_{\LTL}\theta,\\
\alpha,i\models_{\LTL}\theta_1\U\theta_2
 &\text{ iff } \text{there is } j\geq i \text{ with }
       \alpha,j\models_{\LTL}\theta_2\ \text{and}\
       \alpha,k\models_{\LTL}\theta_1 \text{ for all } i\leq k<j .
\end{align*}
As usual $\F\theta:=\top\U\theta$ and $\G\theta:=\neg\F\neg\theta$.

\begin{defi}[Translation of $\LTL$]\label{def:ltl-translation}
Define $\tau$ recursively by $\tau(p):=p$, $\tau(\top):=\top$,
$\tau(\neg\theta):=\neg\tau(\theta)$,
$\tau(\theta_1\wedge\theta_2):=\tau(\theta_1)\wedge\tau(\theta_2)$,
$\tau(\X\theta):=(\tau(\theta))_1$, and
\begin{equation}\label{eq:ltl-until}
\tau(\theta_1\U\theta_2):=\tau(\theta_1)\U\tau(\theta_2)
\;=\;\tau(\theta_2)\vee\exists n\,
 \bigl((\tau(\theta_2))_n\wedge\boldBox_n\tau(\theta_1)\bigr),
\end{equation}
where $\U$ on the right-hand side is the operator of
Section~\ref{subsec:linear}, with a fresh length variable $n$.
\end{defi}

\noindent
The formula $\tau(\theta)$ is closed in $\ELN$. The first disjunct
of~\eqref{eq:ltl-until} covers a witness at the current position; a positive
length cannot reach it. In the second disjunct, the shared length $n$
selects the positions at which the left operand is required. No comparison
of terms is needed.

\begin{prop}[Embedding of $\LTL$]\label{prop:ltl}
For every formula $\theta$ of $\LTL$, every $\alpha\colon\Nat\to2^{\Prop}$ and
every position $i$,
\begin{equation}\label{eq:ltl-correctness}
\alpha,i\models_{\LTL}\theta
\quad\Longleftrightarrow\quad
(\alpha,i)\models\tau(\theta).
\end{equation}
The translation is computable in linear time when common subformulas are
represented once. It preserves and reflects satisfiability, validity and
equivalence.
\end{prop}
\begin{proof}
We proceed by induction on $\theta$, simultaneously for all positions. The
next case follows from displacement by $1$, the until case from
Lemma~\ref{lem:until}, and the Boolean cases from their satisfaction clauses.
When common subformulas are represented once, a constructor of $\theta$
adds a constant number of nodes. Without this sharing, repeated copies of
$\tau(\theta_2)$ in~\eqref{eq:ltl-until} may cause exponential growth when
right operands are nested. Since~\eqref{eq:ltl-correctness} holds for
arbitrary $\alpha$ and $i$, the last claim follows.
\end{proof}

\noindent
For a fresh point variable $x$ the translation gives the unbounded modalities
of Section~\ref{subsec:linear},
\[
\tau(\F\theta)\equiv\exists x\,(\tau(\theta))_x
\qquad\text{and}\qquad
\tau(\G\theta)\equiv\forall x\,(\tau(\theta))_x ,
\]
since $\boldBox_n\top$ is valid and Lemma~\ref{lem:sorts} brings witnesses
at the current position and at positive distance under one point quantifier. For instance, $\G(p\rightarrow\F q)$ becomes, up to
equivalence, $\forall x\,\bigl(p_x\rightarrow\exists y\,q_{x+y}\bigr)$. Both
variables are point variables: a response may occur at the position of the
request, and its delay may depend on that position. \emph{Strict until}
$\theta_1\Us\theta_2$ requires a witness $j>i$ and tests $\theta_1$ at the
positions $i<k<j$ only. It is expressed by $\X(\theta_1\U\theta_2)$ and is
translated by $(\tau(\theta_1)\U\tau(\theta_2))_1$.

By the classical correspondence between $\LTL$ at position $0$, first-order
logic of order, and the star-free
$\omega$-languages~\cite{kamp,gpss,thomas79}, Proposition~\ref{prop:ltl} places
the star-free $\omega$-languages among the $\ELN$-definable ones. In the
alphabetic presentation, letters take the place of atomic propositions.
The inclusion is proper (Section~\ref{sec:hierarchy}). Beyond the image of
$\tau$, variables need not follow the pattern of Lemma~\ref{lem:until}. A length may be shared among several windows, as in
the Introduction, or a variable may occur twice in a parameter, as in
$\forall x\,a_{2x}$.

\section{Satisfiability and Post's Correspondence Problem}\label{sec:pcp}

We now prove the main theorem and establish its sharpness. In the Introduction
a shared length compared window lengths; here it compares their contents.
We read an $\omega$-word as a table whose row width is fixed by one shared
length. Modal operators then impose conditions within rows and between
consecutive rows. A second shared length bounds the table, and a window of
that length imposes the conditions throughout it.

\begin{thm}[Main theorem]\label{thm:main}
Satisfiability for $\ELN$ is undecidable, already for closed formulas of the
form $\exists N\,\exists n\,\Psi$ over a fixed alphabet of thirteen letters,
where $N$ and $n$ are length variables and $\Psi$ is built from letters by
$\wedge$, $\vee$, displacement and bounded universality.
\end{thm}

The formulas in the theorem contain neither universal quantifiers nor
negation. Bounded universality expresses the universal conditions, and over
$\omega$-words the complement of a set of letters is again a set of letters.
The proof reduces Post's Correspondence Problem (PCP): a formula
$\Phi_\Delta$ of the stated form has a model if and only if a given instance
$\Delta$ has a solution (Sections~\ref{subsec:pcp}--\ref{subsec:correct}).
Proposition~\ref{prop:ledgers} characterizes its models. Section~\ref{subsec:sharp}
determines the degree of the problem, shows that the theorem is sharp, and
treats universal quantifiers and validity. Degrees of unsolvability refer to
the arithmetical hierarchy: $\Sigma^0_1$ is the class of the recursively
enumerable sets and $\Pi^0_1$ that of their complements; hardness and
completeness refer to computable many-one reductions.

\subsection{Post's Correspondence Problem}\label{subsec:pcp}

\begin{defi}[PCP]\label{def:pcp}
Let $\Gamma$ be a finite alphabet. An \emph{instance} of PCP over $\Gamma$ is a
finite nonempty list $\Delta=\bigl((u_1,v_1),\ldots,(u_d,v_d)\bigr)$ of pairs of
nonempty words over $\Gamma$, called \emph{dominoes}, with \emph{top} words
$u_i$ and \emph{bottom} words $v_i$. A \emph{solution} of $\Delta$ is a
nonempty sequence of indices $i_1,\ldots,i_m$ such that
$u_{i_1}u_{i_2}\cdots u_{i_m}=v_{i_1}v_{i_2}\cdots v_{i_m}$; this word is its
\emph{common word}.
\end{defi}

\noindent
Following the usual notation for PCP, $u$, $v$ and $s$ denote words in this
section, $d$ the number of dominoes of an instance, $i$ the index of a domino,
and $m$ the number of indices of a solution. Whether an instance of PCP has
a solution is undecidable~\cite{pcp,sipser}. The halting problem reduces to
PCP, and solutions can be enumerated, so PCP is $\Sigma^0_1$-complete. This holds already for
$\Gamma=\{0,1\}$: coding the $i$-th letter of a larger alphabet by the binary
word $10^i$, for $i\geq1$, preserves and reflects equality of concatenations.
Small instances can have very long shortest solutions~\cite{binarypcp2}. Our
running example is the instance
$\Delta_0=\bigl((1,111),\ (10111,10),\ (10,0)\bigr)$, which has the solution
$2,1,1,3$:
\[
\underbrace{10111}_{u_2}\,\underbrace{1}_{u_1}\,\underbrace{1}_{u_1}\,\underbrace{10}_{u_3}
\;=\;101111110\;=\;
\underbrace{10}_{v_2}\,\underbrace{111}_{v_1}\,\underbrace{111}_{v_1}\,\underbrace{0}_{v_3}.
\]

\subsection{The ledger}
The construction begins with an observation: a shared length copies the
contents of a window. The formula $\boldBox_n(a\rightarrow a_n)$ requires
each occurrence of $a$ in the window of length $n$ to recur $n$ positions
later. This carries a word from one row of a table to the next. The remaining
conditions must be checked within a row or between consecutive rows. We call
the resulting table a ledger.

Let $i_1,\ldots,i_m$ be a solution with common word $s$. The solution can be
verified by reading $s$ with two pointers, one for the top words and one for
the bottom words, both initially at the first letter. A domino $(u_i,v_i)$ is
\emph{applicable} if $u_i$ occurs in $s$ at the top pointer and $v_i$ occurs
at the bottom pointer, and applying it advances the pointers by $|u_i|$ and
$|v_i|$, respectively. A sequence is a solution if and only if each domino is
applicable in turn and both pointers finish immediately after $s$. The ledger
is the record of this verification: row $r$ contains a copy of $s$ and the
two pointer positions after $r$ dominoes have been applied. Rows have the
same width $L$. In the concatenated $\omega$-word, corresponding cells of
consecutive rows are $L$ positions apart. Displacement by $L$ thus compares
a cell with the one directly below it.

\begin{defi}[Ledger alphabet]\label{def:alphabet}
Choose distinct letters $\cent$ and $\$$ outside $\Gamma$, and put
$\Gamma_{\$}:=\Gamma\cup\{\$\}$. A \emph{cell} is a triple
$\langle a,\mu,\nu\rangle\in\Gamma_{\$}\times\{0,1\}\times\{0,1\}$, with
\emph{base letter} $a$; it \emph{carries the top mark} if $\mu=1$ and
\emph{the bottom mark} if $\nu=1$, and it is a \emph{terminal cell} if
$a=\$$. We display a cell as a two-storey symbol with $a$ on both storeys, the
upper one dotted if the cell carries the top mark and the lower one dotted if
it carries the bottom mark; thus $\cl{\mk{1}}{1}$ is the cell
$\langle1,1,0\rangle$. The \emph{ledger alphabet} is
\[
\Sigma_\Gamma:=\{\cent\}\;\cup\;\bigl(\Gamma_{\$}\times\{0,1\}\times\{0,1\}\bigr),
\]
which has $4(|\Gamma|+1)+1$ letters: thirteen for $\Gamma=\{0,1\}$,
independently of the instance. The letter $\cent$ and the cells that carry no
mark are the \emph{unmarked} letters. We use the letter classes
\[
[a]:=\{\langle a,\mu,\nu\rangle\mid\mu,\nu\in\{0,1\}\}\ \ (a\in\Gamma_{\$}),\qquad
[\Gamma]:=\textstyle\bigcup_{a\in\Gamma}[a],
\]
\[
\Tm:=\{\langle a,1,\nu\rangle\mid a\in\Gamma_{\$},\ \nu\in\{0,1\}\},\qquad
\Bm:=\{\langle a,\mu,1\rangle\mid a\in\Gamma_{\$},\ \mu\in\{0,1\}\},
\]
and, for a word $w=a_1\cdots a_r\in\Gamma^{+}$, the word formula
$[w]:=\iota^{[a_1]\cdots[a_r]}$ of Section~\ref{subsec:words}, which holds at a
position exactly when the cells from that position onwards spell $w$,
whatever marks they carry (Lemma~\ref{lem:word}).
\end{defi}

\begin{defi}[Rows and the ledger of a solution]\label{def:ledger}
A position of an $\omega$-word over $\Sigma_\Gamma$ that carries $\cent$ is a
\emph{row start}. Given $N\geq1$ and $L:=N+2$, the $L$ letters from a row
start onwards form a \emph{row}, of \emph{width} $L$, and the letter of a row
at distance $o$ from its row start is at \emph{offset} $o$; a number
$r\in\Nat$ is a \emph{row index} if $rL$ is a row start, and \emph{row $r$} is
then the row that begins at $rL$. Now let $i_1,\ldots,i_m$ be a solution of
$\Delta$ whose common word $s=s_1\cdots s_N$ has length $N$, and put
$s_{N+1}:=\$$. For $0\leq r\leq m$ let
\[
j_r:=1+|u_{i_1}\cdots u_{i_r}|\qquad\text{and}\qquad k_r:=1+|v_{i_1}\cdots v_{i_r}|
\]
be the offsets of the top and of the bottom pointer after $r$ dominoes, so
that $j_0=k_0=1$ and $j_m=k_m=N+1$. Put
$R_r:=\cent\,c^{r}_1c^{r}_2\cdots c^{r}_{N+1}$, where
$c^{r}_o:=\langle s_o,\mu,\nu\rangle$ with $\mu=1$ exactly when $o=j_r$, and
$\nu=1$ exactly when $o=k_r$. The \emph{ledger} of the solution is the finite
word $R_0R_1\cdots R_m$. In an $\omega$-word that begins with the ledger, the
numbers $0,\ldots,m$ are row indices, row $r$ is $R_r$, and $c^r_o$ is at
position $rL+o$.
\end{defi}

\noindent
Figure~\ref{fig:ledger} displays the ledger of the running example with one
row per line.

\begin{figure}[htbp]
\[
\renewcommand{\arraystretch}{2.1}
\begin{array}{r@{\quad}c@{\;\;}*{10}{c}@{\qquad}l}
 & & \scriptstyle 1 & \scriptstyle 2 & \scriptstyle 3 & \scriptstyle 4 & \scriptstyle 5 &
     \scriptstyle 6 & \scriptstyle 7 & \scriptstyle 8 & \scriptstyle 9 & \scriptstyle 10 & \\
R_0 & \cent & \cL{\mk{1}}{\mk{1}} & \cL00 & \cL11 & \cL11 & \cL11 & \cL11 & \cL11 & \cL11 & \cL00 & \cL{\$}{\$}
   & \text{\small apply } (u_2,v_2)=(10111,\,10)\\
R_1 & \cent & \cL11 & \cL00 & \cL{1}{\mk{1}} & \cL11 & \cL11 & \cL{\mk{1}}{1} & \cL11 & \cL11 & \cL00 & \cL{\$}{\$}
   & \text{\small apply } (u_1,v_1)=(1,\,111)\\
R_2 & \cent & \cL11 & \cL00 & \cL11 & \cL11 & \cL11 & \cL{1}{\mk{1}} & \cL{\mk{1}}{1} & \cL11 & \cL00 & \cL{\$}{\$}
   & \text{\small apply } (u_1,v_1)=(1,\,111)\\
R_3 & \cent & \cL11 & \cL00 & \cL11 & \cL11 & \cL11 & \cL11 & \cL11 & \cL{\mk{1}}{1} & \cL{0}{\mk{0}} & \cL{\$}{\$}
   & \text{\small apply } (u_3,v_3)=(10,\,0)\\
R_4 & \cent & \cL11 & \cL00 & \cL11 & \cL11 & \cL11 & \cL11 & \cL11 & \cL11 & \cL00 & \cL{\mk{\$}}{\mk{\$}}
   & \text{\small both marks on the terminal cell}
\end{array}
\]
\caption{The ledger of the solution $2,1,1,3$ of $\Delta_0$. Here
$s=101111110$, $N=9$ and $L=11$, and the ledger is obtained by concatenating
$R_0R_1R_2R_3R_4$. Dotted red upper entries indicate top marks, at offsets
$1,6,7,8,10$; dotted lower entries indicate bottom marks, at offsets
$1,3,6,9,10$. Vertically aligned cells are $L$ positions apart,
and the last row starts at position $4L=44$.}\label{fig:ledger}
\end{figure}

\subsection{The formula}
Two difficulties shape the reduction formula. Tests must remain within a
row, although the language cannot compare an offset with the row width
(Remark~\ref{rem:predicate-free}). The alphabet enforces this: an occurrence
of a word over $\Gamma$ cannot include a terminal cell. The conditions must
also hold throughout the table without a universal quantifier. A second
length specifies the start of the last row; the window of that length at
position $0$ contains the preceding rows.

\begin{defi}[The formula $\Phi_\Delta$]\label{def:phi-delta}
Let $N$ and $n$ be length variables, and let $L$ abbreviate the term $N+2$.
For a letter class $A\subseteq\Sigma_\Gamma$, with each letter identified
with its singleton class, we read $\neg A$ as the letter class
$\Sigma_\Gamma\setminus A$, to which it is equivalent in the alphabetic
presentation, and $A\rightarrow\chi$ accordingly as
$(\Sigma_\Gamma\setminus A)\vee\chi$. Define the following formulas:
\begin{align*}
\Init  &:= \cent\wedge(\Tm\wedge\Bm)_1,\\[2pt]
\Frame &:= \cent\rightarrow\bigl((\boldBox_N[\Gamma])_{1}\wedge[\$]_{N+1}
           \wedge\cent_{L}\bigr),\\[2pt]
\Copy  &:= \bigwedge_{a\in\Gamma}\bigl([a]\rightarrow[a]_{L}\bigr),\\[2pt]
\Step  &:= \cent\rightarrow\bigvee_{i=1}^{d}\Bigl(\boldBox_{N+1}\bigl(
   (\Tm\rightarrow([u_i]\wedge\Tm_{L+|u_i|}))\wedge
   (\Bm\rightarrow([v_i]\wedge\Bm_{L+|v_i|}))\bigr)\Bigr)_{1},\\[2pt]
\Last  &:= \cent\wedge\bigl(\boldBox_N(\neg\Tm\wedge\neg\Bm)\bigr)_{1}.
\end{align*}
The reduction formula is
\[
\Phi_\Delta:=\exists N\,\exists n\,\Psi_\Delta
\qquad\text{with}\qquad
\Psi_\Delta:=\Init\wedge\boldBox_n(\Frame\wedge\Copy\wedge\Step)\wedge\Last_n .
\]
\end{defi}

\noindent
For the lengths assigned to $N$ and $n$, the components act as follows.
$\Init$ requires a row to start at position $0$, with both marks on its
first cell. The window $\boldBox_n$ imposes $\Frame$, $\Copy$ and $\Step$
at positions before $n$. The formula $\Frame$ requires rows starting there
to have the form $\cent\,\Gamma^N\$$, ignoring marks, and to be followed by
another row start. The formula $\Copy$ requires a nonterminal cell to have
the same base letter as the cell $L$ positions later. The verification is
carried by $\Step$, which selects one domino applicable at all marks of a
row. For some $i$, each top-marked cell begins an occurrence of $u_i$, and a
top mark occurs $|u_i|$ cells further to the right in the next row. The
analogous condition holds for bottom marks and $v_i$. In particular,
terminal cells of these rows cannot carry a mark. Finally, $\Last$ requires
$n$ to be a row start and the first $N$ cells of that row to be unmarked. The
marks placed in that row by $\Step$ must therefore occur on its terminal
cell, so both readings end together. Uniqueness of each kind of mark within a row follows
(Remark~\ref{rem:marks}). The formula
$\Phi_\Delta$ is computable from $\Delta$ in polynomial time.

\subsection{Correctness}\label{subsec:correct}

\begin{lem}[Completeness]\label{lem:complete}
Let $i_1,\ldots,i_m$ be a solution of $\Delta$ with common word $s$. Every
$\omega$-word over $\Sigma_\Gamma$ that begins with the ledger of the solution
is a model of $\Phi_\Delta$; the matrix $\Psi_\Delta$ holds at position $0$
when $N$ denotes $|s|$ and $n$ denotes $m(|s|+2)$.
\end{lem}
\begin{proof}
Let $\sigma$ begin with the ledger, with $s$, $N$, $L$, $j_r$ and $k_r$ as in
Definition~\ref{def:ledger}, and assign $|s|$ to the variable $N$ and $mL$ to
$n$. Then $\Init$ holds at position $0$: $\sigma[0]=\cent$, and
$\sigma[1]=c^0_1$ carries both marks since $j_0=k_0=1$. The positions before
$mL$ are those of the rows $0,\ldots,m-1$. Among them, the antecedents of
$\Frame$ and of $\Step$ hold exactly at the row starts, and those of $\Copy$
at the cells of $[\Gamma]$. We check these conditions in turn.

$\Frame$, at a row start $rL$ with $r<m$: the positions $rL+1,\ldots,rL+N$
carry the cells $c^r_1,\ldots,c^r_N$, which are in $[\Gamma]$; position
$rL+N+1$ carries $c^r_{N+1}\in[\$]$; and $\sigma[(r+1)L]=\cent$ because
$r+1\leq m$.

$\Copy$, at a position $rL+o$ with $r<m$ that carries a cell of $[a]$,
$a\in\Gamma$: here $o\leq N$ and $s_o=a$, and the letter $L$ positions later
is $c^{r+1}_o\in[a]$.

$\Step$, at a row start $rL$ with $r<m$: let $i:=i_{r+1}$. The window of the
$i$-th disjunct consists of the positions of the cells of row $r$. The only
one of these cells with a top mark is $c^r_{j_r}$. Since
$s=u_{i_1}\cdots u_{i_m}$, the word $u_i$ occurs in $s$ at offset $j_r$.
Thus $[u_i]$ holds at $rL+j_r$ by Lemma~\ref{lem:word}. The position
$rL+j_r+L+|u_i|=(r+1)L+j_{r+1}$ carries the top mark in row $r+1$.
The bottom-mark condition follows in the same way, using $k_r$ and $v_i$.

$\Last$, at the position $mL$, which carries $\cent$: the cells
$c^m_1,\ldots,c^m_N$ are unmarked since $j_m=k_m=N+1$.
\end{proof}

\begin{lem}[Soundness]\label{lem:sound}
If $\Phi_\Delta$ is satisfiable, then $\Delta$ has a solution.
\end{lem}
\begin{proof}
By Lemma~\ref{lem:suffix}, $\Phi_\Delta$ has a model $\sigma$. Fix witnesses
for $N$ and $n$, again denoted $N$ and $n$, and let $L=N+2$. Rows and row
indices are understood as in Definition~\ref{def:ledger} for this value of
$N$. The formulas $\Frame$, $\Copy$ and $\Step$ hold at positions before $n$.

\smallskip\noindent
\emph{(1) Rows.} By $\Init$, $0$ is a row index, and $0<n$ since $n$ is a
length. Let $r$ be a row index with $rL<n$. By $\Frame$ at $rL$, the letters $\sigma[rL+o]$ are in
$[\Gamma]$ for $1\leq o\leq N$, the letter $\sigma[rL+N+1]$ is in $[\$]$, and
$r+1$ is a row index. By $\Last$, $\sigma[n]=\cent$; thus $n$ cannot be any of the
positions $rL+1,\ldots,rL+N+1$, since these carry cells. It follows that
$(r+1)L\leq n$. Induction shows that an $r$ with $rL<n$ is a row index with
$(r+1)L\leq n$. For the least $m$ with $mL\geq n$, this gives $m\geq1$ and
$n=mL$. Thus $0,\ldots,m$ are row indices, and rows $0,\ldots,m-1$ have the
shape just described. For $r<m$ let $s^r\in\Gamma^N$ be the word spelled by the
cells of row $r$ at the offsets $1,\ldots,N$.

\smallskip\noindent
\emph{(2) One word.} If $r+1<m$, then $s^{r}=s^{r+1}$: for $1\leq o\leq N$,
let $a$ be the $o$-th letter of $s^r$; the formula $\Copy$ at position $rL+o$
requires $\sigma[(r+1)L+o]$ to belong to $[a]$. Let $s=s_1\cdots s_N$ be the
common value of $s^0,\ldots,s^{m-1}$.

\smallskip\noindent
\emph{(3) The trace.} We construct offsets $j_r,k_r\in\{1,\ldots,N+1\}$ for
$0\leq r\leq m$, with $\sigma[rL+j_r]\in\Tm$ and $\sigma[rL+k_r]\in\Bm$, and
indices $i_1,\ldots,i_m$. By $\Init$ we may start with $j_0=k_0=1$. Let $r<m$,
and suppose that $j_r$ and $k_r$ have been chosen. Since $\Step$ holds at
$rL$, fix a disjunct $i$ that holds there. Its window consists of the
positions $rL+1,\ldots,rL+N+1$, including the top-marked position $rL+j_r$. Hence
$[u_i]$ holds at $rL+j_r$, and $\Tm$ at $rL+j_r+L+|u_i|$. By
Lemma~\ref{lem:word} the letters at the $|u_i|$ consecutive positions from
$rL+j_r$ onwards are in $[\Gamma]$ and spell $u_i$. Since $j_r\leq N+1$ and the
cell of row $r$ at offset $N+1$ is terminal by~(1), the occurrence cannot
cross it. Its cells therefore have offsets at most $N$, and~(2) gives
\[
j_{r+1}:=j_r+|u_i|\leq N+1
\qquad\text{and}\qquad
s_{j_r}s_{j_r+1}\cdots s_{j_{r+1}-1}=u_i ,
\]
and $\sigma[(r+1)L+j_{r+1}]\in\Tm$. In the same way, with the same $i$,
$k_{r+1}:=k_r+|v_i|\leq N+1$, $s_{k_r}\cdots s_{k_{r+1}-1}=v_i$ and
$\sigma[(r+1)L+k_{r+1}]\in\Bm$. Put $i_{r+1}:=i$.

\smallskip\noindent
\emph{(4) The last row.} By~(1) and $\Last$, the letters at the offsets
$1,\ldots,N$ of row $m$ are unmarked. The marked offsets from~(3) therefore
satisfy $j_m=k_m=N+1$.

\smallskip\noindent
\emph{(5) The solution.} By (3), $s_{j_{r-1}}\cdots s_{j_r-1}=u_{i_r}$ for
$1\leq r\leq m$; since $j_0=1$ and $j_m=N+1$ these factors tile $s$, so that
$s=u_{i_1}\cdots u_{i_m}$. In the same way $s=v_{i_1}\cdots v_{i_m}$. As
$m\geq1$, the sequence $i_1,\ldots,i_m$ is a solution of $\Delta$.
\end{proof}

\begin{proof}[Proof of Theorem~\ref{thm:main}]
By Lemmas~\ref{lem:complete} and~\ref{lem:sound}, the instance $\Delta$ has a
solution if and only if $\Phi_\Delta$ is satisfiable. Thus
$\Delta\mapsto\Phi_\Delta$ is a computable many-one reduction of PCP to
satisfiability. PCP is undecidable for $\Gamma=\{0,1\}$, whose ledger
alphabet has thirteen letters. The formula $\Phi_\Delta$ is closed and has the
form $\exists N\,\exists n\,\Psi_\Delta$, and, after expanding negations and
implications as specified in Definition~\ref{def:phi-delta}, its matrix is
built from letters by $\wedge$, $\vee$, displacement and bounded
universality. Here letter classes are disjunctions of letters, and word
formulas are conjunctions of displaced letter classes.
\end{proof}

\begin{prop}[Models are ledgers]\label{prop:ledgers}
Let $\sigma$ be an $\omega$-word over $\Sigma_\Gamma$, and let $N$ and $n$
denote lengths. Then $\Psi_\Delta$ holds at position $0$ of $\sigma$ if and
only if $\Delta$ has a solution $i_1,\ldots,i_m$ with a common word of length
$N$ such that $n=m(N+2)$ and $\sigma$ begins with
$R_0\cdots R_{m-1}\,\cent\,c^m_1\cdots c^m_N\,\langle a,1,1\rangle$ for some
$a\in\Gamma_{\$}$, that is, with the ledger of this solution up to the base
letter of its last cell.
\end{prop}
\begin{proof}
If $\sigma$ has this form, the proof of Lemma~\ref{lem:complete} applies
without change, since it uses the last cell of the ledger only through its
two marks. Conversely, let $\Psi_\Delta$ hold at position $0$ of $\sigma$.
Steps (1)--(5) of the proof of Lemma~\ref{lem:sound} give $n=mL$, the word
$s$ of length $N$ spelled by the rows $0,\ldots,m-1$, and a solution
$i_1,\ldots,i_m$ with common word $s$ whose pointers have the offsets $j_r$
and $k_r$. By $\Last$ the letter at position $mL$ is $\cent$, and by $\Copy$
at the positions $(m-1)L+o$ the cell of row $m$ at an offset $o\leq N$ has
the base letter $s_o$. To prove uniqueness of the top mark, we show that the
cell at offset $j_r$ is the only top-marked cell of row $r$, for
$0\leq r\leq m$. Bottom marks are treated likewise. For $r=m$ this is
step~(4).
If $r<m$ and the cell of row $r$ at offset $o$ carries a top mark, then the
argument of step~(3), for the disjunct $i=i_{r+1}$ and with $o$ in place of
$j_r$, places a top mark on the cell of row $r+1$ at offset
$o+|u_i|\leq N+1$. Induction on $m-r$ gives $o+|u_i|=j_{r+1}=j_r+|u_i|$,
hence $o=j_r$. Thus row $r$ is $R_r$ for $r<m$, and row $m$ agrees with $R_m$ up
to offset $N$ and has a cell with both marks at offset $N+1$.
\end{proof}

\noindent
A further conjunct $[\$]_{n+N+1}$ would determine the remaining base letter.

\begin{rem}[Uniqueness of marks]\label{rem:marks}
For a letter class $A\in\{\Tm,\Bm\}$, the formula
\[
\One_A:=\bigl(\cent\rightarrow(\mdiamond_{N+1}A)_1\bigr)\wedge
  \bigl(A\rightarrow(\boldBox_N\neg A)_1\bigr)
\]
requires a row to contain a cell of $A$, and the $N$ letters after a cell
of $A$ to lie outside $A$. A row has $N+1$ cells, so exactly one is then in
$A$. This property is implicit in $\Psi_\Delta$: we have
$\Psi_\Delta\models\boldBox_n(\One_\Tm\wedge\One_\Bm)$. Indeed, by
Proposition~\ref{prop:ledgers} the rows are those of a ledger: row $r$ has
its only top mark at offset $j_r$, and as $j_r<j_{r+1}$, the $N$ letters
after it carry no top mark if $r<m$; bottom marks are treated likewise. The second
conjunct of $\One_A$ has the form admitted in Theorem~\ref{thm:main};
adjoining it to the window of $\Psi_\Delta$ yields an equivalent formula.
\end{rem}

\subsection{Degree, sharpness, and universal quantifiers}\label{subsec:sharp}
The formulas of Theorem~\ref{thm:main} are existential in the sense of
Section~\ref{subsec:laws}. Whenever an existential formula holds, a finite
prefix of the $\omega$-word suffices to verify it. This determines the degree
of unsolvability. Table~\ref{tab:degrees} collects the results.

\begin{table}[htbp]
\centering
\small
\renewcommand{\arraystretch}{1.12}
\begin{tabular}{@{}l@{\qquad}l@{\qquad}l@{}}
\hline\noalign{\vskip 2pt}
Closed formulas & Satisfiability & Source\\[1pt]
\hline\noalign{\vskip 2pt}
existential formulas without $\boldBox_t$ & NP-complete & \cite{DZ26}\\
$\exists\xi_1\cdots\exists\xi_k\,\Psi$, with $\Psi$ lacking one construct
  & decidable & Proposition~\ref{prop:necessary}\\
$\exists\xi\,\Psi$, with $\Psi$ quantifier-free & decidable & Proposition~\ref{prop:onevar}\\
$\exists N\,\exists n\,\Psi$, with $\Psi$ using all four constructs
  & $\Sigma^0_1$-complete & Theorem~\ref{thm:main}, Corollary~\ref{cor:prefix}\\
existential formulas & $\Sigma^0_1$-complete & Corollary~\ref{cor:prefix}\\
$\forall n\,\Psi$, with $\Psi$ using $\wedge$, $\vee$ and displacement
  & $\Pi^0_1$-complete & Proposition~\ref{prop:universal}\\
all formulas & $\Sigma^1_1$-complete & \cite{DZ26}, Theorem~\ref{thm:analytic}\\[1pt]
\hline
\end{tabular}
\caption{Degrees of satisfiability for closed formulas of $\ELN$. In the
second, fourth and sixth rows the matrix $\Psi$ is built from letters; the
\emph{four constructs} are $\wedge$, $\vee$, displacement and bounded
universality.}
\label{tab:degrees}
\end{table}

\begin{cor}[Existential formulas]\label{cor:prefix}
Satisfiability of existential formulas of $\ELN$ is $\Sigma^0_1$-complete. It
is $\Sigma^0_1$-hard for the formulas of Theorem~\ref{thm:main}, and in the
propositional presentation for the formulas $\exists N\,\exists n\,\Psi$ in
which $\Psi$ is built from literals over four atomic propositions by $\wedge$,
$\vee$, displacement and bounded universality.
\end{cor}
\begin{proof}
Let $\varphi$ be an existential formula with $(\sigma,\eta,i)\models\varphi$.
By induction on $\varphi$ there is a number $h$, a \emph{bound}, such that
$\varphi$ holds at $i$ under $\eta$ in every $\omega$-word that agrees with
$\sigma$ at all positions below $h$, even when the quantifiers of $\varphi$
are restricted to values below $h$; every number above a bound is a bound. For
a literal take $h:=i+1$; for $\wedge$, $\vee$, displacement and
$\mdiamond_t$ take the largest of the bounds of the operands that hold, at the
positions concerned; for $\boldBox_t\psi$ take the largest of the bounds of
$\psi$ at the finitely many positions of the window; and for
$\exists\xi\,\psi$ with witness $e$ take a bound of $\psi$ under
$\eta[\xi\mapsto e]$ that exceeds $e$. With Lemma~\ref{lem:suffix} it
follows that a closed existential formula $\varphi$ is satisfiable if and only
if there are a number $h$ and a word $w$ of length $h$ such that $\varphi$
holds at position $0$ of $wa^\omega$, for a fixed letter $a$, with its
quantifiers restricted to values below $h$. For fixed $h$ and $w$, this
condition is decidable. Enumerating finite witnesses therefore enumerates the
satisfiable existential formulas. An open formula is satisfiable if and only
if its existential closure is. Hardness for the fragment in the main theorem
follows from the many-one reduction $\Delta\mapsto\Phi_\Delta$ of PCP. In
the propositional presentation, let $\Gamma=\{0,1\}$ and let two atomic
propositions code the symbols
$\cent,0,1,\$$ by their truth values, and let two more state the marks of a
cell. A set of these propositions codes one letter of $\Sigma_\Gamma$,
with four codes for $\cent$. In $\Phi_\Delta$, replace letters by
disjunctions of the conjunctions of literals describing their codes. The
resulting formula has the stated form and holds under an interpretation exactly when
$\Phi_\Delta$ holds in the decoded $\omega$-word.
\end{proof}

\noindent
The second quantified length is necessary. With a single quantified variable
$N$, a formula inspects boundedly many consecutive segments of length $N$.
These form a finite word over a product alphabet, making satisfiability
decidable.

\begin{prop}[One quantified variable]\label{prop:onevar}
Satisfiability is decidable for the closed formulas $\exists\xi\,\Psi$ of
$\ELN$ in which $\Psi$ is quantifier-free.
\end{prop}
\begin{proof}
The truth of a closed quantifier-free formula at position $0$ depends on a
prefix of computable length, so its satisfiability is decidable. By
Lemma~\ref{lem:sorts}, it remains to treat a length variable $N$ in place of
$\xi$. Every parameter in $\Psi$ then has an affine form $kN+c$ with numerals
$k$ and $c$. Let $c^{*}$ be the largest of the numerals $c$, and $K-1$ the sum
of the numerals $k$ and $c$ of all parameters in $\Psi$. We first handle the finitely many values
$N\leq c^{*}$ by substitution. For $N>c^{*}$ the truth of $\Psi$ at
position $0$ depends only on the positions below $KN$, since a modal operator
advances the position of evaluation by at most the value of its parameter.
Write these positions as $rN+o$ with $r<K$ and $o<N$, and let the
\emph{column word} of $\sigma$ be the word of length $N$ over $\Sigma^{K}$
whose $o$-th letter lists $\sigma[rN+o]$ for $r<K$. Displacement by $kN+c$
leads from $(r,o)$ to $(r+k,o+c)$ if $o+c<N$ and to $(r+k+1,o+c-N)$ otherwise,
and a window is an interval in the lexicographic order of these pairs. As $K$,
$k$ and $c$ do not depend on $N$, and $c<N$, the truth of each subformula of
$\Psi$ at $rN+o$ is expressed, by induction, by a first-order formula in $o$
over the column word. Hence $\Psi$ holds at position $0$ of some
$\omega$-word for some $N>c^{*}$ if and only if a first-order sentence,
computed from $\Psi$, holds in some word over $\Sigma^{K}$ of length greater
than $c^{*}$. This is expressible in the monadic second-order theory of
$(\Nat,<)$, which is decidable~\cite{buchi62}.
\end{proof}

\noindent
The four constructs are also necessary, regardless of the number of
quantified variables. Thus the fragment of Theorem~\ref{thm:main} is minimal;
the proof identifies their roles.

\begin{prop}[Necessity of the constructs]\label{prop:necessary}
Satisfiability is decidable for the closed formulas
$\exists\xi_1\cdots\exists\xi_k\,\Psi$ of $\ELN$ in which $\Psi$ is built from
letters by only three of the four constructs $\wedge$, $\vee$, displacement
and bounded universality.
\end{prop}
\begin{proof}
By Lemma~\ref{lem:suffix} it suffices to consider position $0$. \emph{Without
bounded universality}, the formula is existential without $\boldBox_t$; such
formulas translate into existential Presburger arithmetic, and their
satisfiability is NP-complete~\cite{DZ26}. \emph{Without $\wedge$}, an
induction on $\Psi$ shows that $\Psi$ holds at every position of $a^\omega$
for each letter $a$ that occurs in it. \emph{Without displacement}, if $\Psi$
holds at a position that carries the letter $a$, then it holds at every
position of $a^\omega$, again by induction, because a window contains its
point of evaluation. In $a^\omega$ the modal operators have no effect, so
it suffices to evaluate $\Psi$ there for each letter $a$. \emph{Without $\vee$},
Propositions~\ref{prop:disp} and~\ref{prop:modal}(ii) turn $\Psi$ into a
conjunction of formulas $(\boldBox_{t_1}\cdots\boldBox_{t_r}a)_u$ with
$r\geq0$. By Proposition~\ref{prop:linear}(v), such a formula holds at
position $0$ exactly when $a$ occupies the $t_1+\cdots+t_r-(r-1)$ consecutive
positions from the position denoted by $u$ onwards. For given values of the
variables, the conjunction has a model if and only if the intervals assigned
to distinct letters are disjoint, a quantifier-free condition of Presburger
arithmetic. Satisfiability is therefore expressed by an existential sentence
of Presburger arithmetic.
\end{proof}

A single universal quantifier also yields undecidability, without bounded
modalities. Satisfiability is then co-recursively enumerable and complete for
that class.

\begin{prop}[One universal quantifier]\label{prop:universal}
Satisfiability is $\Pi^0_1$-complete for the closed formulas $\forall n\,\Psi$
of $\ELN$ with a length variable $n$ and with $\Psi$ built from letters by
$\wedge$, $\vee$ and displacement.
\end{prop}
\begin{proof}
By Lemma~\ref{lem:suffix} it suffices to consider position $0$. For a fixed
value of $n$, the truth of $\Psi$ depends on a prefix of computable length of
the $\omega$-word. Given a number $b$, it is therefore decidable whether some
$\omega$-word satisfies $\Psi$ at position $0$ for all values of $n$ up to
$b$. Compactness of $\Sigma^\omega$ implies that $\forall n\,\Psi$ is
satisfiable if and only if this holds for every $b$. The problem is
therefore in $\Pi^0_1$. The same argument applies to any prefix of universal
quantifiers followed by a quantifier-free matrix.

For hardness, let $T$ be a finite set of \emph{tiles} with sets
$H,V\subseteq T\times T$ of admissible horizontal and vertical pairs; a
\emph{tiling of the quadrant} is a map $\theta\colon\Nat^2\to T$ with
$(\theta(i,j),\theta(i+1,j))\in H$ and $(\theta(i,j),\theta(i,j+1))\in V$ for
all $i,j$. Whether $T$ tiles the quadrant, or equivalently the plane, is
undecidable~\cite{berger1966}. A reduction of the complement of the halting
problem establishes $\Pi^0_1$-completeness. Take $T$ as the
alphabet, and let $\Psi_T$ be the disjunction
\[
\Psi_T:=\bigvee\bigl\{\,a_n\wedge b_{2n}\wedge c_{3n}\bigm|
   (a,b)\in H,\ (a,c)\in V\,\bigr\},
\]
which we may assume to be nonempty, as $T$ does not tile the quadrant
otherwise. If $\forall n\,\Psi_T$ holds at position $0$ of $\sigma$, then
$\theta(i,j):=\sigma[2^i3^j]$ is a tiling of the quadrant: for $n=2^i3^j$ the
disjunct that holds puts $(\sigma[n],\sigma[2n])$ in $H$ and
$(\sigma[n],\sigma[3n])$ in $V$. Conversely, a tiling $\theta$ yields the
model $\sigma$ with $\sigma[n]:=\theta(i,j)$ for $n\geq1$, where $2^i$ and
$3^j$ are the highest powers of $2$ and of $3$ that divide $n$.
\end{proof}

\noindent
In the propositional presentation one atomic proposition $p$ suffices: for
tiles $t_0,\ldots,t_{K-1}$, let $p$ at position $Kn+k$ state that the
tile at $n$ is $t_k$, and let $\Psi_T$ require $k$ to be unique.

\begin{cor}[Validity]\label{cor:nonre}
The set of valid formulas of $\ELN$ is not recursively enumerable, and neither
is the set of valid quantifier-free formulas, which is $\Pi^0_1$-complete. Hence no proof system with a
decidable notion of proof is sound and complete for $\ELN$.
\end{cor}
\begin{proof}
By Lemmas~\ref{lem:complete} and~\ref{lem:sound}, $\Delta$ has no solution
exactly when the quantifier-free formula $\neg\Psi_\Delta$ is valid. An
enumeration of the valid quantifier-free formulas would thus enumerate the
instances of PCP without a solution. Since those with a solution are
recursively enumerable, PCP would be decidable. A sound and complete proof system with a
decidable notion of proof would provide such an enumeration. By the same
reduction the set of valid quantifier-free formulas is $\Pi^0_1$-hard. It
is in $\Pi^0_1$ because such a formula is valid exactly when its negation, an
existential formula, is unsatisfiable (Corollary~\ref{cor:prefix}).
\end{proof}

\section{The Expressive Power of \texorpdfstring{$\ELN$}{EL(N)}}\label{sec:hierarchy}

The main theorem concerns a small fragment of $\ELN$. We now classify
satisfiability for the full logic in the analytical hierarchy and compare the
$\ELN$-definable languages with classical classes of $\omega$-languages. The
standard translation, the degree of unsolvability, and parts (i) and (ii) of
the hierarchy theorem are due to Droste and the author~\cite{DZ26}. We present
them here in two-sorted notation. Example~\ref{ex:wolper}, part (iii) of the hierarchy theorem, and
the results on B\"uchi automata in Section~\ref{subsec:buchi} complete the
comparison.

\subsection{The standard translation}
Let $\FO(\Nat,<,+;\Sigma)$ be first-order logic over the structure
$(\Nat,<,+)$ expanded by a unary predicate $P_a$ for each letter $a\in\Sigma$;
an $\omega$-word $\sigma$ interprets $P_a$ as
$P_a^\sigma:=\{i\mid\sigma[i]=a\}$. Let $\FO(\Nat,<;\Sigma)$ be its fragment
without addition. The terms of $\FO(\Nat,<,+;\Sigma)$ are built from variables
and constants for natural numbers by addition, and $u\leq v$ abbreviates
$u<v\vee u=v$. The variables of $\ELN$, of both sorts, are among the variables
of $\FO(\Nat,<,+;\Sigma)$, where all variables range over $\Nat$.

\begin{thm}[Standard translation~\cite{DZ26}]\label{thm:forward}
There is a translation, computable in linear time when common subterms are
represented once, that assigns to every formula $\varphi$ of $\ELN$ and every
variable $z$ of $\FO(\Nat,<,+;\Sigma)$ that does not occur in $\varphi$ a
formula $T_z(\varphi)$ of $\FO(\Nat,<,+;\Sigma)$ such that, for every
$\omega$-word $\sigma$, assignment $\eta$ and position $i$,
\[
(\sigma,\eta,i)\models\varphi
\quad\Longleftrightarrow\quad
\bigl(\Nat,<,+,(P_a^{\sigma})_{a\in\Sigma}\bigr),\,\eta[z\mapsto i]\models T_z(\varphi).
\]
\end{thm}
\begin{proof}
Define $T_e(\varphi)$ more generally for a term $e$ of
$\FO(\Nat,<,+;\Sigma)$, after renaming the bound variables of $\varphi$ to avoid
those of $e$ and choosing a fresh variable $y$ for each bounded modality:
\begin{align*}
T_e(a)&:=P_a(e), & T_e(\varphi_u)&:=T_{e+u}(\varphi),\\
T_e(\neg\varphi)&:=\neg T_e(\varphi), &
T_e(\varphi\circ\psi)&:=T_e(\varphi)\circ T_e(\psi)\quad(\circ\in\{\wedge,\vee\}),\\
T_e(\mdiamond_t\varphi)&:=\exists y\,(e\leq y\wedge y<e+t\wedge T_y(\varphi)), &
T_e(\boldBox_t\varphi)&:=\forall y\,(e\leq y\wedge y<e+t\rightarrow T_y(\varphi)),\\
T_e(\exists x\,\varphi)&:=\exists x\,T_e(\varphi), &
T_e(\exists n\,\varphi)&:=\exists n\,(1\leq n\wedge T_e(\varphi)),
\end{align*}
and dually for the universal quantifiers. Induction on $\varphi$ shows that
$T_e(\varphi)$ holds under $\eta$ exactly when $\varphi$ holds at position
$\eta(e)$. The displacement clause applies the induction hypothesis to $e+u$.
The bounded-modality clauses describe the window $[\eta(e),\eta(e)+\eta(t))$. The theorem is the case in which $e$ is the
variable $z$. With common subterms represented once, the translation is
computed in linear time.
\end{proof}

\noindent
For a formula in ADNF (Theorem~\ref{thm:adnf}), the translation is simple:
a displaced letter $a_u$ becomes an atomic formula $P_a(e+u)$, and a bounded
modality becomes a quantifier bounded by terms. The translation also gives
an upper bound on the degree of unsolvability. A formula is satisfiable if and
only if some partition of $\Nat$ into sets $P_a$ and some values of its
variables, positive for length variables, make its translation true.
This bound is attained.

\begin{thm}[Analytical classification~\cite{DZ26}]\label{thm:analytic}
Satisfiability for $\ELN$ is $\Sigma^1_1$-complete, and validity is
$\Pi^1_1$-complete.
\end{thm}

\noindent
The lower bound in~\cite{DZ26} is obtained by reduction from the recurrence
problem for two-counter machines. Such a nondeterministic machine has finite
control and two counters, which it may increment, decrement and test for
zero. The problem asks whether it has an infinite computation that visits a
designated control state infinitely often~\cite{harel1983}.

\subsection{Language inclusions and separations}
A \emph{B\"uchi automaton} $\mathcal{A}=(Q,\Sigma,\delta,q_0,F)$ has a finite
set $Q$ of states, a transition relation
$\delta\subseteq Q\times\Sigma\times Q$, an initial state $q_0$ and a set
$F\subseteq Q$ of accepting states. A \emph{run} on an $\omega$-word $\sigma$
is a sequence $q_0q_1\cdots$ of states with $(q_i,\sigma[i],q_{i+1})\in\delta$
for all $i$. It is \emph{accepting} if $q_i\in F$ for infinitely many $i$.
Write $\Lang(\mathcal{A})$ for the set of $\omega$-words with an accepting run. An
$\omega$-language is \emph{$\omega$-regular} if it is $\Lang(\mathcal{A})$ for
some B\"uchi automaton $\mathcal{A}$~\cite{buchi62,Tho90}. We use two
classical facts: runs of B\"uchi automata can be pumped, and every
$\omega$-language definable in $\FO(\Nat,<;\Sigma)$ is $\omega$-regular.

\begin{exa}[Equal blocks]\label{ex:abcd}
Over $\Sigma=\{a,b,c,d\}$, the formula
\[
 \varphi_{abcd}:=\exists n\,\bigl(\boldBox_n a\wedge(\boldBox_n b)_n
       \wedge(\boldBox_n c)_{2n}\wedge(\boldBox_n d)_{3n}\bigr)
\]
defines $L_{abcd}:=\{a^nb^nc^nd^n\mid n\geq1\}\Sigma^\omega$~\cite{DZ26}: the
four consecutive windows share the length $n$, and each consists of a
\emph{block}, a factor containing repetitions of a single letter. This language
is not $\omega$-regular. Suppose a B\"uchi automaton with fewer than $m$ states
accepted it. While reading the first $m$ letters of $a^mb^mc^md^\omega$, an
accepting run repeats a state. Pumping the corresponding loop
lengthens the block of $a$'s while leaving the block of $b$'s unchanged.
The automaton would then accept an $\omega$-word outside $L_{abcd}$. The same
argument applies to the language of the Introduction.
\end{exa}

\begin{exa}[Iterated blocks]\label{ex:blocks}
The language $L_{blk}:=\{a^nb^nc^n\mid n\geq1\}^{\omega}$, whose definability
was posed in~\cite{zhang2024temporal} (see also~\cite{oregular}), is defined over
$\{a,b,c\}$ by
\begin{align*}
 \Blk&:=\exists n\,\bigl(\boldBox_na\wedge(\boldBox_nb)_n
                      \wedge(\boldBox_nc)_{2n}\wedge a_{3n}\bigr),\\
 \varphi_{blk}&:=\Blk\wedge\forall x\,\bigl((c\wedge a_1)_x\rightarrow\Blk_{x+1}\bigr).
\end{align*}
The formula $\Blk$ states that three blocks of equal length, of $a$'s, $b$'s
and $c$'s, are followed by $a$. The first conjunct of $\varphi_{blk}$ places
such a triple at position $0$. A later boundary between $c$ and $a$ begins
another; by induction the triples exhaust the $\omega$-word.
Conversely, every such boundary in an $\omega$-word of $L_{blk}$ begins a
triple. Pumping the first block of $a$'s of $(a^mb^mc^m)^\omega$ shows that
$L_{blk}$ is not $\omega$-regular. Unlike in Example~\ref{ex:abcd}, the
witness $n$ is chosen afresh at each boundary: an existentially quantified
length depends on a universally quantified position
(Remark~\ref{rem:prenex}).
\end{exa}

\begin{exa}[Parity~\cite{DZ26}]\label{ex:parity}
For a finite word $w$, let $|w|_a$ count the occurrences of $a$ in $w$.
Let $\Sigma=\{a,b,c,d\}$ and
\[
L_{par}:=\{w\,c\,d^{\omega}\mid w\in\{a,b\}^{*},\ |w|_a\text{ even}\}.
\]
This language is $\omega$-regular. It is not definable in
$\FO(\Nat,<,+;\Sigma)$, and hence, by Theorem~\ref{thm:forward}, not
$\ELN$-definable~\cite{DZ26}. This follows from the classical lower bound for parity: no family of Boolean
circuits of constant depth and polynomial size, with gates of unbounded
fan-in, computes parity~\cite{fss1984,ajtai1983,hastad1986}. To outline the
argument, suppose a closed formula $\psi$ of $\FO(\Nat,<,+;\Sigma)$ defined $L_{par}$. On an
$\omega$-word $wcd^\omega$ with $|w|=n$, replace the atomic formulas $P_c(e)$,
$P_d(e)$, $P_a(e)$ and $P_b(e)$ of $\psi$ by $e=n$, $e>n$,
$\bigvee_{i<n}(e=i\wedge X_i)$ and $\bigvee_{i<n}(e=i\wedge\neg X_i)$, where
the Boolean variable $X_i$ states that $w$ has the letter $a$ at position $i$.
The result is a formula of Presburger arithmetic with a fixed number of
quantifiers and $O(n)$ atomic formulas. Cooper's quantifier
elimination~\cite{cooper1972}, applied to it without expansion into
disjunctive normal form, replaces each quantified variable by polynomially
many substitution instances and adds a bounded number of layers of
connectives. This yields circuits of constant depth and polynomial size that
decide the parity of $X_0,\ldots,X_{n-1}$, a contradiction.
\end{exa}

\begin{exa}[Wolper's property]\label{ex:wolper}
Let $\Sigma=\{a,b\}$ and let $L_{even}$ be the set of $\omega$-words that carry
$a$ at every even position. It is defined by $\forall x\,a_{2x}$, and it is
$\omega$-regular. It is not star-free: Wolper showed that this property is not
expressible in $\LTL$~\cite{wolper1983}, and the $\LTL$-definable
$\omega$-languages are the star-free ones.
\end{exa}

\begin{thm}[Hierarchy]\label{thm:hierarchy}
Fix a finite alphabet $\Sigma$ with at least four letters. Write
$\Lang(\mathcal{C})$ for the class of $\omega$-languages over $\Sigma$ that
are definable in the logic $\mathcal{C}$, and $\oReg$ for the class of
$\omega$-regular languages over $\Sigma$.
\begin{enumerate}[label=\textup{(\roman*)},leftmargin=2.2em]
\item $\Lang(\FO(\Nat,<;\Sigma))\subsetneq\Lang(\ELN)\subseteq\Lang(\FO(\Nat,<,+;\Sigma))$
\textup{\cite{DZ26}}.
\item $\Lang(\ELN)$ and $\oReg$ are incomparable \textup{\cite{DZ26}}.
\item $\Lang(\FO(\Nat,<;\Sigma))\subsetneq\Lang(\ELN)\cap\oReg$.
\end{enumerate}
\end{thm}
\begin{proof}
(i) The first inclusion is Proposition~\ref{prop:ltl} together with Kamp's
theorem~\cite{kamp,gpss}; it is proper by Example~\ref{ex:abcd}, since
$\FO(\Nat,<;\Sigma)$ defines only $\omega$-regular languages. The second
inclusion is Theorem~\ref{thm:forward}. (ii) Incomparability follows from Examples~\ref{ex:abcd}
and~\ref{ex:parity}. (iii) The inclusion follows from (i), and it is proper by
Example~\ref{ex:wolper}. An example over fewer letters transfers to
$\Sigma$ by a conjunct confining all positions to the smaller alphabet.
\end{proof}

\noindent
Figure~\ref{fig:hierarchy} summarizes the theorem and the examples.

\begin{figure}[htbp]
\centering
\begin{tikzpicture}[font=\small]
\draw[line width=0.8pt] (0,0) ellipse [x radius=5.4, y radius=2.75];
\draw[inkblue,line width=1pt,fill=inkblue!7] (-0.7,-0.1) ellipse [x radius=3.7, y radius=1.95];
\draw[inkred,line width=1pt] (3.0,-0.35) ellipse [x radius=4.1, y radius=1.35];
\draw[inkgreen,line width=1pt,fill=inkgreen!12] (0.55,-0.35) ellipse [x radius=1.25, y radius=0.62];
\node at (0,2.35){$\Lang(\FO(\Nat,<,+;\Sigma))$};
\node[inkblue] at (-1.6,1.45){$\Lang(\ELN)$};
\node[inkred] at (6.5,1.1){$\oReg$};
\node[inkgreen!60!black,align=center] at (0.55,-0.35){\scriptsize star-free\\[-2pt]\scriptsize $\LTL$ languages};
\fill (-3.0,0.35) circle (1.6pt) node[right=1pt]{\scriptsize $L_{abcd}$};
\fill (-2.7,-0.55) circle (1.6pt) node[right=1pt]{\scriptsize $L_{blk}$};
\fill (2.25,0.42) circle (1.6pt) node[above=1pt]{\scriptsize $L_{even}$};
\fill (6.25,-0.45) circle (1.6pt) node[below=1pt]{\scriptsize $L_{par}$};
\node at (3.35,1.7){\large ?};
\end{tikzpicture}
\caption{Schematic comparison of the $\ELN$-definable languages with the
classical language classes in Theorem~\ref{thm:hierarchy}. The marked
languages are those of Examples~\ref{ex:abcd}--\ref{ex:wolper}. The question
mark indicates Conjecture~\ref{conj:proper}: some language definable with
addition is not $\ELN$-definable.}
\label{fig:hierarchy}
\end{figure}

\begin{conj}[Proper inclusion~\cite{DZ26}]\label{conj:proper}
The inclusion $\Lang(\ELN)\subseteq\Lang(\FO(\Nat,<,+;\Sigma))$ is proper.
\end{conj}

\noindent
The obstacle is described in Remark~\ref{rem:predicate-free}. A formula
evaluates a letter at the position denoted by an affine term. The syntax,
however, offers no comparison between the positions denoted by two terms;
first-order logic with addition can express such a comparison by an equation.
A separation proof therefore needs to exclude indirect encodings of these
comparisons through the letters of the $\omega$-word. The results here do not
depend on the conjecture. Section~\ref{sec:discussion} returns to it in
connection with automata for $\ELN$.

\subsection{B\"uchi automata and projections}\label{subsec:buchi}
Although $\ELN$ does not define all $\omega$-regular languages, it defines the
accepting runs of every B\"uchi automaton
$\mathcal{A}=(Q,\Sigma,\delta,q_0,F)$. Let $\Sigma_{\mathcal{A}}:=Q\times\Sigma$,
write $(q/a)$ for its letters, and let
$\pi\colon\Sigma_{\mathcal{A}}^{\omega}\to\Sigma^{\omega}$ be the
letter-to-letter \emph{projection} $(q/a)\mapsto a$. An \emph{annotated run} is
an $\omega$-word $(q_0/\sigma[0])(q_1/\sigma[1])\cdots$ for a run
$q_0q_1\cdots$ on $\sigma$.

\begin{prop}[Runs are definable]\label{prop:buchi}
For every B\"uchi automaton $\mathcal{A}$ the formula
\[
\varphi_{\mathcal{A}}:=\bigvee_{a\in\Sigma}(q_0/a)\;\wedge\;
\forall x\bigwedge_{(q,a)\in Q\times\Sigma}\Bigl((q/a)_x\rightarrow
  \bigvee_{\substack{q'\in Q,\ b\in\Sigma\\(q,a,q')\in\delta}}(q'/b)_{x+1}\Bigr)\;\wedge\;
\forall y\,\exists x\bigvee_{q\in F,\;a\in\Sigma}(q/a)_{x+y}
\]
defines the set of annotated accepting runs of $\mathcal{A}$. Hence
$\pi(\Lang(\varphi_{\mathcal{A}}))=\Lang(\mathcal{A})$.
\end{prop}
\begin{proof}
The first conjunct fixes the initial state. The second requires consecutive
letters to respect $\delta$. The third requires, at or beyond any position
$y$, a position $x+y$ whose state belongs to $F$.
\end{proof}

\begin{cor}[Projections]\label{cor:projection}
Every $\omega$-regular language is the image of an $\ELN$-definable language
under a letter-to-letter projection. The class of $\ELN$-definable languages
is closed under Boolean operations and under the operations
$L\mapsto\Sigma^{k}L$ and $L\mapsto\Sigma^{*}L$, but it is not closed under
projection.
\end{cor}
\begin{proof}
The first claim is Proposition~\ref{prop:buchi}. Boolean closure follows from
the connectives. By Lemma~\ref{lem:suffix},
$\Lang(\varphi_k)=\Sigma^k\Lang(\varphi)$ and
$\Lang(\exists x\,\varphi_x)=\Sigma^{*}\Lang(\varphi)$. If the class were
closed under projection it would contain
$L_{par}=\pi(\Lang(\varphi_{\mathcal{A}}))$ for an automaton $\mathcal{A}$
accepting $L_{par}$, contradicting Example~\ref{ex:parity}.
\end{proof}

\noindent
The formula $\varphi_{\mathcal{A}}$ uses only displacement by $1$ and the
unbounded modalities. It is equivalent to the translation of a formula of
$\LTL$. The first claim of the corollary thus recovers the classical fact
that the $\omega$-regular languages are projections of star-free ones.
The failure of closure under projection distinguishes the two classes.
Projection makes an existential choice of sets of positions; the quantifiers
of $\ELN$ range only over positions and lengths. Shared lengths nonetheless
suffice to express equality of block lengths.

\section{An Equational Calculus}\label{sec:calculus}

By Corollary~\ref{cor:nonre}, no sound proof system with a decidable notion of
proof is complete for $\ELN$, even for quantifier-free formulas. A sound
calculus remains useful for rewriting specifications. We give an equational
calculus that derives window composition laws, an induction principle for
windows, and an atomic-displacement normal form without negation. Most axioms
instantiate the semantic laws in Sections~\ref{sec:general}
and~\ref{sec:eln}. Throughout this section, we use
the alphabetic presentation over a fixed alphabet $\Sigma$.

Judgments are equations $\varphi=\psi$ between formulas. We use
$\varphi\leq\psi$ to abbreviate $\varphi\vee\psi=\psi$. We write $u\doteq v$ if the terms $u$ and
$v$ have the same affine form, or equivalently, if they denote the same number
under every assignment. The calculus has reflexivity ($\varphi=\varphi$),
symmetry, transitivity, and replacement of equals in a formula context---a
formula with one designated subformula occurrence. Its remaining axioms and
rules follow; $\xi$ and $\zeta$ stand for variables of either sort.

\begin{description}[leftmargin=1.2em,style=nextline,itemsep=3pt]
\item[$\bm{(\Sigma)}$]
$\bigwedge_{a\in A}\neg a=\bigvee(\Sigma\setminus A)$ for every
$A\subseteq\Sigma$, with $\bigvee\emptyset=\bot$ and $\bigwedge\emptyset=\top$.
\item[(Bool)] All instances of the equations of Boolean algebras for $\wedge$,
$\vee$, $\neg$, $\bot$ and $\top$.
\item[(Quant)]
For an affine term $u$ and a positive affine term $t$, with substitution as
in Section~\ref{subsec:syntax}:
\begin{itemize}[leftmargin=2.6em,itemsep=0pt,topsep=2pt]
\item[(Q1)] $\exists\xi\,\varphi=\varphi$ if $\xi$ is not free in $\varphi$;
\item[(Q2)] $\exists\xi\,(\varphi\vee\psi)=\exists\xi\,\varphi\vee\exists\xi\,\psi$;
\item[(Q3)] $\exists\xi\,(\varphi\wedge\psi)=\exists\xi\,\varphi\wedge\psi$ if
$\xi$ is not free in $\psi$;
\item[(Q4)] $\exists\xi\,\exists\zeta\,\varphi=\exists\zeta\,\exists\xi\,\varphi$;
\item[(Q5)] $\varphi[x:=u]\leq\exists x\,\varphi$ and
$\varphi[n:=t]\leq\exists n\,\varphi$;
\item[(Q6)] $\exists\xi\,\varphi=\exists\zeta\,\varphi[\xi:=\zeta]$ if $\zeta$
has the sort of $\xi$ and does not occur in $\varphi$.
\end{itemize}
The corresponding laws for $\forall$ follow by (Neg).
\item[(Neg)]
$\neg(\varphi_u)=(\neg\varphi)_u$,\;
$\neg\boldBox_t\varphi=\mdiamond_t\neg\varphi$,\;
$\neg\mdiamond_t\varphi=\boldBox_t\neg\varphi$,\;
$\neg\forall\xi\,\varphi=\exists\xi\,\neg\varphi$,\;
$\neg\exists\xi\,\varphi=\forall\xi\,\neg\varphi$.
\item[(Disp)]
$\varphi_0=\varphi$,\; $(\varphi_u)_v=\varphi_{u+v}$,\;
$(\varphi\wedge\psi)_u=\varphi_u\wedge\psi_u$,\;
$(\varphi\vee\psi)_u=\varphi_u\vee\psi_u$,\;
$(\boldBox_t\varphi)_u=\boldBox_t(\varphi_u)$,\;
$(\mdiamond_t\varphi)_u=\mdiamond_t(\varphi_u)$,\; and
$(Q\xi\,\varphi)_u=Q\xi\,(\varphi_u)$ if $\xi$ does not occur in $u$.
\item[(Term)]
$\varphi_u=\varphi_v$ if $u\doteq v$;\;
$\boldBox_s\varphi=\boldBox_t\varphi$ and $\mdiamond_s\varphi=\mdiamond_t\varphi$
if $s\doteq t$.
\item[(Win)]
For positive affine terms $s,t$:\;
$\boldBox_1\varphi=\varphi$,\;
$\boldBox_{s+t}\varphi=\boldBox_s\varphi\wedge(\boldBox_t\varphi)_s$,\;
$\boldBox_t(\varphi\wedge\psi)=\boldBox_t\varphi\wedge\boldBox_t\psi$;\;
$\mdiamond_1\varphi=\varphi$,\;
$\mdiamond_{s+t}\varphi=\mdiamond_s\varphi\vee(\mdiamond_t\varphi)_s$,\;
$\mdiamond_t(\varphi\vee\psi)=\mdiamond_t\varphi\vee\mdiamond_t\psi$.
\item[(Mono)]
From $\varphi\leq\psi$ infer $\varphi_u\leq\psi_u$,\;
$\boldBox_t\varphi\leq\boldBox_t\psi$,\;
$\mdiamond_t\varphi\leq\mdiamond_t\psi$,\;
$\forall\xi\,\varphi\leq\forall\xi\,\psi$ and
$\exists\xi\,\varphi\leq\exists\xi\,\psi$.
\item[(Sorts)]
$\exists n\,\varphi=\exists x\,\varphi[n:=x+1]$ and
$\exists x\,\varphi=\varphi[x:=0]\vee\exists n\,\varphi[x:=n]$, where the
variable introduced on the right does not occur in $\varphi$.
\item[(Ind)]
$\varphi[n:=1]\wedge\forall n\,(\varphi\rightarrow\varphi[n:=n+1])\leq\forall n\,\varphi$.
\end{description}

\noindent
The axioms $(\Sigma)$ are particular to $\omega$-words. The axiom
$\boldBox_1\varphi=\varphi$ expresses the fact that a window of length $1$
contains a single position; it has no counterpart in dense time. The axiom
(Ind) yields laws uniformly in a length variable; separate derivations give
only their numeral instances. Generalization follows from (Mono) and (Q1),
and instantiation from (Q5) and (Neg). Derivable equations are therefore closed
under capture-avoiding substitution for free variables. The axioms
(Sorts) distinguish the position $0$ from a positive position, and the length
$1$ from a successor.

\begin{thm}[Soundness]\label{thm:sound}
If $\varphi=\psi$ is derivable then $\varphi\equiv\psi$, and if
$\varphi\leq\psi$ is derivable then $\varphi\models\psi$.
\end{thm}
\begin{proof}
The axioms are valid equivalences. For $(\Sigma)$ this is
Definition~\ref{def:words}. The axioms (Bool) and (Quant) hold as in classical
two-sorted logic, since both sorts are nonempty. Propositions~\ref{prop:disp}
and~\ref{prop:modal} give (Neg), (Disp) and (Term), since the future mode on
$\Nat$ is covariant. The axiom (Win) follows from
Proposition~\ref{prop:linear}(i),(v) and Proposition~\ref{prop:modal}(i),(ii);
(Sorts) from Lemma~\ref{lem:sorts}; and (Ind) by induction on $\Natp$.
Semantic equivalence is a congruence and the constructs in (Mono) are monotone
(Proposition~\ref{prop:modal}(iv)), so the rules preserve validity.
\end{proof}

\begin{prop}[Derived laws]\label{prop:derived}
The following are derivable, for all formulas $\varphi,\psi$, affine terms $u$
and positive affine terms $s,t$.
\begin{enumerate}[label=\textup{(\roman*)},leftmargin=2.2em]
\item $\boldBox_t\varphi\leq\varphi\leq\mdiamond_t\varphi$.
\item $\boldBox_{s+t}\varphi\leq\boldBox_s\varphi$ and
$\mdiamond_s\varphi\leq\mdiamond_{s+t}\varphi$.
\item \textup{(Reach)}\; $\boldBox_{u+1}\varphi\leq\varphi_u$ and
$\varphi_u\leq\mdiamond_{u+1}\varphi$.
\item \textup{(Numerals)}\; $\boldBox_k\varphi=\bigwedge_{0\leq i<k}\varphi_i$
and $\mdiamond_k\varphi=\bigvee_{0\leq i<k}\varphi_i$ for every numeral
$k\geq1$.
\item \textup{(Distribution)}\;
$\boldBox_t(\varphi\rightarrow\psi)\leq\boldBox_t\varphi\rightarrow\boldBox_t\psi$
and
$\boldBox_t(\varphi\rightarrow\psi)\leq\mdiamond_t\varphi\rightarrow\mdiamond_t\psi$.
\item \textup{(Composition)}\;
$\boldBox_{s+1}\boldBox_t\varphi=\boldBox_{s+t}\varphi$ and
$\mdiamond_{s+1}\mdiamond_t\varphi=\mdiamond_{s+t}\varphi$; hence
$\boldBox_s\boldBox_t\varphi=\boldBox_t\boldBox_s\varphi$ and
$\mdiamond_s\mdiamond_t\varphi=\mdiamond_t\mdiamond_s\varphi$.
\item \textup{(Bounded induction)}\;
$\boldBox_t(\varphi\rightarrow\varphi_1)\leq\varphi\rightarrow\boldBox_{t+1}\varphi$.
\end{enumerate}
\end{prop}
\begin{proof}
It suffices to derive the laws for $\boldBox$; those for $\mdiamond$ follow by
(Neg). We derive each law for a fresh variable and then instantiate it. For (i), (Sorts) reduces a length to $1$ or to $n+1$, and
$\boldBox_{n+1}\varphi=\varphi\wedge(\boldBox_n\varphi)_1\leq\varphi$ by (Win).
For (iii), (Sorts) reduces a position to $0$ or to a length $n$, and
$\boldBox_{n+1}\varphi=\boldBox_n\varphi\wedge\varphi_n$. Parts (ii) and (iv)
follow directly from (Win). Part (v) follows by applying (Mono) to
$(\varphi\rightarrow\psi)\wedge\varphi\leq\psi$ and, for its second
inequality, to $\varphi\leq\psi\vee\neg(\varphi\rightarrow\psi)$.

For (vi), (Win) and (i) give
$\boldBox_{s+1}\boldBox_t\varphi\leq\boldBox_{s+t}\varphi$. For the converse, define $\chi(n):=\boldBox_{n+t}\varphi\rightarrow\boldBox_n\boldBox_t\varphi$,
with $n$ fresh for $\varphi$ and $t$. The case $n=1$ is
$\boldBox_{t+1}\varphi\leq\boldBox_t\varphi$. Under $\chi(n)$, the formula
$\boldBox_{n+1+t}\varphi$ implies both $\boldBox_n\boldBox_t\varphi$ and
$(\boldBox_t\varphi)_n$, the latter by decomposition at $n$. Their conjunction
is $\boldBox_{n+1}\boldBox_t\varphi$, giving the induction step. Apply
(Ind), instantiate $n:=s$, and adjoin $(\boldBox_t\varphi)_s$ to obtain the
converse inequality. Commutation follows by reducing each length to $1$ or to
a successor and using composition and the commutativity of addition.

For (vii), apply (Ind) to
\[
 \chi(n):=\bigl(\boldBox_n(\varphi\rightarrow\varphi_1)\wedge\varphi\bigr)
                              \rightarrow\boldBox_{n+1}\varphi,
\]
where $n$ is fresh for $\varphi$. The base case is modus ponens. For the
step, decompose the window of length $n+1$ at its end: under
$\chi(n)$ the antecedent of $\chi(n+1)$ implies
$\boldBox_{n+1}\varphi\wedge(\varphi_n\rightarrow\varphi_{n+1})$, and since
$\boldBox_{n+1}\varphi\leq\varphi_n$ this implies $\boldBox_{n+2}\varphi$.
Generalize and instantiate $n:=t$.
\end{proof}

\begin{thm}[Positive atomic-displacement normal form]\label{thm:positive}
For every formula $\varphi$ there is a formula $\psi$ in ADNF without
negation, that is, a formula built from letters and displaced letters, and
from $\top$ and $\bot$ taken as constants, by $\wedge$, $\vee$, the bounded
modalities and the quantifiers, such that $\varphi=\psi$ is derivable.
\end{thm}
\begin{proof}
By (Neg), (Bool) and replacement of equals, push negations inwards until
they apply only to letters. Then replace $\neg a$ by
$\bigvee(\Sigma\setminus\{a\})$ using $(\Sigma)$. The rewriting of
Theorem~\ref{thm:adnf} is now carried out in the calculus: each step is an
instance of (Disp), with $\bot_u=\bot$ by (Neg), (Disp) and (Bool) when
$|\Sigma|=1$, and (Q6) renames bound variables.
\end{proof}

\noindent
The absence of negation is particular to $\omega$-words, where exactly one
letter holds at each position and the complement of a letter is a finite
disjunction of letters. Universal quantifiers and
bounded universality remain in the normal form, so no expressive power is lost.

The companion Hilbert system~\cite{DZ26} is complete relative to an oracle for
the validity of Presburger arithmetic with uninterpreted unary predicates.
By Theorem~\ref{thm:analytic}, no effective procedure can replace this oracle.
The present calculus has effective, sound rules and serves the more limited
purpose of rewriting.

\section{The Rational Logic \texorpdfstring{$\ELQ$}{EL(Q)}: Interval Ensembles and
Model Checking}\label{sec:qel}

The applications of $\EL$ described in the Introduction call for evaluation
of a given formula in a given interpretation. For the rational logic on
labelled interval ensembles, this problem is tractable in the size of the
data. Modal operators are computed by arithmetic on interval endpoints, and
quantifiers admit finite evidence. A fixed formula can be evaluated in
polynomial time. Satisfiability remains undecidable in dense time
(Theorem~\ref{thm:lineartime}), already for existential formulas with two
shared lengths (Corollary~\ref{cor:threelengths}). The
presentation adapts the framework of~\cite{qel} to the notation of
Section~\ref{sec:general}, treating both sorts of quantifiers explicitly.
Statements marked ``cf.''\ are adapted from that work.

\subsection{The logic \texorpdfstring{$\ELQ$}{EL(Q)}}

\begin{conv}[The logic $\ELQ$]\label{conv:elq}
Throughout this section the logic is the rational logic $\ELQ$ of
Definition~\ref{def:tel}, called \emph{Rational Ensemble Logic} (QEL)
in~\cite{qel}. Its points are the rational numbers, its lengths the positive
rationals, its window mode the future mode, and its constants the rationals;
Convention~\ref{conv:coercion} applies.
\end{conv}

\noindent
The future mode is covariant and $\Rat$ is dense, so the displacement laws
and the composition law of dense time apply (Propositions~\ref{prop:disp}
and~\ref{prop:linear}(iv)). The calculus of Section~\ref{sec:calculus}
depends on the unit of discrete time and is specific to $\ELN$. Displacement by a negative point allows reference to the
past, although the windows are directed towards the future. We retain both
sorts of variables: every length is $\norm{x}$ for some point $x$, but a
length variable may occur in the parameter of a displacement, whereas
$\norm{x}$ is not an affine term. For example, with atomic propositions $S$,
$T$, $C$ for a seizure and its tonic and clonic activity, the specification
$\exists x\,\exists a\,\exists b\,(\boldBox_{a+b}S\wedge\mdiamond_aT\wedge(\mdiamond_bC)_a)_x$
of~\cite{qel}, with a point variable $x$ and length variables $a,b$, requires a
seizure throughout a window of length $a+b$, with tonic activity in its
first $a$ time units and clonic activity in the remaining $b$.

\subsection{Labelled interval ensembles}
An interpretation over $\Rat$ has a finite representation when atomic
propositions have supports that are finite unions of intervals.

\begin{defi}[Intervals and interval ensembles]\label{def:ie}
An \emph{interval} of $\Rat$ is a nonempty convex subset of $\Rat$ whose
endpoints lie in $\Rat\cup\{-\infty,\infty\}$; a finite endpoint may be
included (closed) or excluded (open), and singletons $[q,q]$ are intervals. An
\emph{interval ensemble} is a finite set $\delta$ of pairwise disjoint
intervals, with union $\bigcup\delta$. It is \emph{canonical} if no union of
two of its intervals is an interval; every finite union of intervals is the
union of exactly one canonical interval ensemble. We write $\IE$ for the set
of interval ensembles.
\end{defi}

\begin{defi}[Labelled interval ensemble]\label{def:best}
Let $\Prop$ be a finite nonempty set of atomic propositions. A \emph{labelled
interval ensemble} (LINE)~\cite{qel} is a
map $\varepsilon\colon\Prop\to\IE$. It presents the interpretation
\[
\alpha_\varepsilon\colon\Rat\to2^{\Prop},\qquad
\alpha_\varepsilon(x):=\{p\in\Prop\mid x\in\textstyle\bigcup\varepsilon(p)\},
\]
and we write $(\varepsilon,q)\models\varphi$ for
$(\alpha_\varepsilon,q)\models\varphi$, and $\sem{\varphi}_\varepsilon$ for the
support $\sem{\varphi}_{\alpha_\varepsilon}$ of a closed $\varphi$.
\end{defi}

\noindent
Thus $\sem{p}_\varepsilon=\bigcup\varepsilon(p)$. A LINE has the form of an
annotated record, in which segments of a recording are labelled with events
or states~\cite{qel}.

\subsection{Semantics as temporal morphology}
On a LINE, the satisfaction clauses of Definition~\ref{def:sat} become
operations on supports. They are the operations of mathematical
morphology~\cite{matheron1975,serra1982}, directed towards the future.

\begin{defi}[Future-directed morphological operations]\label{def:morph}
Let $A\subseteq\Rat$, $q\in\Rat$ and $r\in\Rat_{>0}$; put
$H_r:=W^{+}(0,r)=[0,r)\cap\Rat$ and $A-q:=\{a-q\mid a\in A\}$. Define
\[
A_q:=\{x\mid x+q\in A\},\qquad
\mdiamond_rA:=\{x\mid(H_r+x)\cap A\neq\emptyset\},\qquad
\boldBox_rA:=\{x\mid H_r+x\subseteq A\},
\]
and $A^{c}:=\Rat\setminus A$; here the modal symbols denote operations on
sets. The first operation is a translation, $A_q=A-q$. In morphological
terms, $\mdiamond_rA=\bigcup_{h\in H_r}(A-h)$ is the
\emph{dilation} of $A$ by the reflected window $(-r,0]$, and
$\boldBox_rA=\bigcap_{h\in H_r}(A-h)$ is the \emph{erosion} of $A$ by $H_r$;
the two are dual, $\boldBox_rA=(\mdiamond_r(A^{c}))^{c}$.
\end{defi}

\noindent
For an interval $I$ with endpoints $a\leq b$, the three operations reduce to
arithmetic on endpoints. The translate $I_q$ has endpoints $a-q$ and $b-q$,
with the same open or closed status as in $I$. The dilation $\mdiamond_rI$
has an open left endpoint $a-r$ and retains the right endpoint of $I$ and
its status. The erosion $\boldBox_rI$ retains the left endpoint of $I$ and
its status; its right endpoint is $b-r$ and is closed. If these
endpoint conditions are inconsistent, the erosion is empty. Infinite endpoints are unaffected. For example
(Figure~\ref{fig:morphology}), when $a<b$,
\[
 \mdiamond_r[a,b)=(a-r,b)\qquad\text{and}\qquad
 \boldBox_r[a,b)=\{x\in\Rat\mid a\leq x\leq b-r\}.
\]
When $a\leq b-r$, the erosion contains $b-r$, because $b-r+h<b$ for $h<r$.
In particular, the erosion of $[a,a+r)$ is the singleton $\{a\}$.
Closed endpoints and singletons are therefore needed even with half-open
input intervals. The two endpoint rules differ because the
future window is half-open.

\begin{figure}[htbp]
\centering
\begin{tikzpicture}[font=\small]
\foreach \x/\lab in {1.0/{a-r},2.6/{a},5.4/{b-r},7.0/{b}}{
  \draw[gray!55,densely dotted] (\x,0.35)--(\x,-2.35);
  \node[below] at (\x,-2.4){\scriptsize $\lab$};}
\foreach \y/\lab in {0/{$I=[a,b)$},-1/{$\mdiamond_rI=(a-r,\,b)$},-2/{$\boldBox_rI=[a,\,b-r]$}}{
  \draw[gray!60,->] (0,\y)--(8.2,\y);
  \node[anchor=west] at (8.4,\y){\lab};}
\draw[line width=2.2pt] (2.6,0)--(7.0,0);
\fill (2.6,0) circle (2.4pt); \draw[fill=white] (7.0,0) circle (2.1pt);
\draw[inkblue,line width=2.2pt] (1.0,-1)--(7.0,-1);
\draw[inkblue,fill=white] (1.0,-1) circle (2.1pt) (7.0,-1) circle (2.1pt);
\draw[inkred,line width=2.2pt] (2.6,-2)--(5.4,-2);
\fill[inkred] (2.6,-2) circle (2.4pt) (5.4,-2) circle (2.4pt);
\end{tikzpicture}
\caption{Endpoint operations for the future window $H_r=[0,r)$. Dilation
shifts the left endpoint to the left by $r$ and makes it open. Erosion shifts
the right endpoint to the left by $r$ and makes it closed. Filled endpoints
are included; hollow endpoints are excluded.}
\label{fig:morphology}
\end{figure}

\pagebreak[3]
\begin{prop}[Compositional semantics and closure~\cite{qel}]\label{prop:morph}
Let $\varepsilon$ be a LINE, let $\varphi,\psi$ be closed formulas,
$q\in\Rat$ and $r\in\Rat_{>0}$. Then
\[
\sem{\varphi\wedge\psi}_\varepsilon=\sem{\varphi}_\varepsilon\cap\sem{\psi}_\varepsilon,
\qquad
\sem{\neg\varphi}_\varepsilon=(\sem{\varphi}_\varepsilon)^{c},
\qquad
\sem{\varphi_q}_\varepsilon=(\sem{\varphi}_\varepsilon)_q,
\]
\[
\sem{\mdiamond_r\varphi}_\varepsilon=\mdiamond_r\sem{\varphi}_\varepsilon,
\qquad
\sem{\boldBox_r\varphi}_\varepsilon=\boldBox_r\sem{\varphi}_\varepsilon .
\]
The unions of interval ensembles are closed under all these operations.
Consequently the support of every quantifier-free closed formula under
$\varepsilon$ is the union of an interval ensemble, which is computed from the
interval ensembles $\varepsilon(p)$ by arithmetic on endpoints.
\end{prop}
\begin{proof}
The identities restate the clauses of Definition~\ref{def:sat}. Finite unions
of intervals are closed under the Boolean operations and under translation.
Dilation distributes over unions, so its endpoint rule gives closure under
dilation. Closure under erosion follows by duality. For the endpoint
rule of erosion on an interval $I$, the inclusion $H_r+x\subseteq I$ requires
$x\in I$ and, at a finite right endpoint $b$, that $x+r\leq b$; the point
$x+r$ itself is not tested. The last claim follows by induction on the
formula.
\end{proof}

\noindent
A quantifier-free closed formula is thus evaluated on a LINE by a
syntax-directed computation on sorted endpoint arrays. Truth at a point $q$
is then decided by a membership test. Erosion requires care: it does not
distribute over unions. For example, $[0,1)\cup[1,2)$ is one interval, and a
future window can cross the cut at $1$. Erosion is therefore applied to
canonical interval ensembles, or computed from dilation by duality. For the logic $\EL(\Real^k)$ in
full mode the same clauses, with the open Euclidean ball $\{v\mid\norm{v}<r\}$
in place of the future window, give translation, dilation and erosion in
space; semi-algebraic sets are closed under them~\cite{qel}.

\subsection{Quantifiers and finite evidence}
The preceding operations evaluate formulas with constant parameters.
Quantifiers require a representation of the assignments under which a
formula holds.

A subset of $\Rat^m$ is \emph{rational semilinear} if it is a finite Boolean
combination of sets defined by linear equalities and strict linear
inequalities with rational coefficients. The hyperplanes of such a description
partition $\Rat^m$ into finitely many relatively open convex sets of all
dimensions, the \emph{faces}, on each of which the description has a constant
truth value. For a LINE $\varepsilon$ and a formula $\varphi$, let
$C_{\varepsilon,\varphi}$ be the finite set consisting of the finite endpoints
of the intervals of $\varepsilon$ and of the constants of $\varphi$. For a
tuple $\boldsymbol\xi=(\xi_1,\ldots,\xi_d)$ of distinct variables, let
$D_{\boldsymbol\xi}:=\prod_{i=1}^d D_i$, where $D_i=\Rat$ if $\xi_i$ is a point
variable and $D_i=\Rat_{>0}$ if it is a length variable; for
$\mathbf a\in D_{\boldsymbol\xi}$ we write $\varphi[\mathbf a]$ for $\varphi$
under an assignment that maps $\boldsymbol\xi$ to $\mathbf a$. Quantification
over a tuple abbreviates successive quantification over its variables.

\begin{thm}[Semilinearity and finite evidence; cf.~\cite{qel}]\label{thm:evidence}
Let $\varepsilon$ be a LINE, and let
$\boldsymbol\xi=(\xi_1,\ldots,\xi_d)$ list the free variables of a formula
$\varphi$ of $\ELQ$.
\begin{enumerate}[label=\textup{(\roman*)},leftmargin=2.2em]
\item The \emph{satisfaction set} of $\varphi$ in $\varepsilon$,
\[
 S_\varphi:=\{(\mathbf a,b)\in D_{\boldsymbol\xi}\times\Rat\mid
                    (\varepsilon,b)\models\varphi[\mathbf a]\},
\]
is rational semilinear and effectively computable. It is described by
relations $k_0y+\sum_i k_i\xi_i+c\sim0$ in the coordinates
$\xi_1,\ldots,\xi_d$ of $\mathbf a$ and $y$ of $b$, where $k_i\in\Int$,
$\sim\,\in\{=,<\}$, and $c$ is an integer combination of elements of
$C_{\varepsilon,\varphi}$.
\item For every $q\in\Rat$ one can compute a finite nonempty set
$F\subseteq D_{\boldsymbol\xi}$, called \emph{finite evidence} at $q$, such that
\begin{align*}
 (\varepsilon,q)\models\exists\boldsymbol\xi\,\varphi
   &\iff (\varepsilon,q)\models\varphi[\mathbf a]
                                  \text{ for some }\mathbf a\in F,\\
 (\varepsilon,q)\models\forall\boldsymbol\xi\,\varphi
   &\iff (\varepsilon,q)\models\varphi[\mathbf a]
                                  \text{ for every }\mathbf a\in F.
\end{align*}
If the section of $S_\varphi$ at $y=q$, together with the conditions
$\xi_i>0$ for the length variables, is described by $g$ hyperplanes, then $F$
may be chosen with one element in each face contained in
$D_{\boldsymbol\xi}$, and $|F|=O((g+1)^d)$ for fixed $d$.
\end{enumerate}
\end{thm}
\begin{proof}
(i) Translate ``$\varphi$ holds at $y$'' into first-order linear arithmetic
over $\Rat$, as in Theorem~\ref{thm:forward}. An atomic proposition $p$
becomes a finite disjunction of constraints placing $y$ in an interval of
$\varepsilon(p)$; a singleton gives an equality. Displacement replaces $y$ by
$y+u$. A bounded modality introduces a fresh variable $y'$ and the constraints
$y\leq y'<y+t$. These appear under an existential quantifier for $\mdiamond$
and as the premise of a universally quantified implication for $\boldBox$. An
occurrence of $\norm{u}$ is resolved by the linear cases $u\geq0$ and $u<0$. Point
quantifiers become quantifiers over $\Rat$, and length quantifiers are
relativized to the positive rationals.

Boolean combinations of rational linear constraints are effectively closed
under projection. Write such a combination as a finite union of systems of equalities and of
strict and nonstrict inequalities. Eliminate equalities by substitution. To
eliminate a variable from the inequalities, pair each lower bound with each
upper bound. The resulting comparison is strict if either bound is strict.
Discard systems with a contradictory constant constraint. As $\Rat$ is dense, the
projection is exact. Universal quantifiers are treated by complementation.
This is quantifier elimination for ordered divisible abelian
groups~\cite{marker2002}. It proves that $S_\varphi$ is rational semilinear
and effectively computable. The elimination steps give the stated form of the defining relations,
though no polynomial bound has yet been established.

(ii) Put $y=q$ and add the hyperplanes $\xi_i=0$ for the length variables.
Membership in $D_{\boldsymbol\xi}$ is constant on each face of the resulting
arrangement, and so is the truth of $\varphi[\mathbf a]$ at $q$ on each face
contained in $D_{\boldsymbol\xi}$. Solve the linear constraints of a face to compute a point in it. Let $F$
contain one point from each of these faces. Since these faces partition the nonempty set
$D_{\boldsymbol\xi}$, the two equivalences follow. An arrangement of $g$
hyperplanes in $\Rat^d$ has $O((g+1)^d)$ faces, by induction on
$g$ and $d$: a new hyperplane is cut by the previous ones into the faces of an
arrangement of dimension $d-1$. For $d=0$, take the empty tuple.
\end{proof}

The hyperplanes in Theorem~\ref{thm:evidence}(ii) arise from translation and
quantifier elimination. Input endpoints alone do not suffice. For example,
if $p$ holds only at $1$, then $p_{x+x}$ holds at $0$ only for $x=1/2$.
The witness lies in a face of dimension $0$ and is not an endpoint of
$\varepsilon$. With alternating quantifiers, the finite evidence for an inner
quantifier depends on the values assigned to the outer variables.

By Theorem~\ref{thm:evidence}(i), the support of a closed formula, including
one with quantifiers, is the union of an interval ensemble: a rational
semilinear subset of $\Rat$ has finitely many boundary points. Intermediate
computations nevertheless use several dimensions. By~(ii), when
$\exists\boldsymbol\xi\,\varphi$ holds at $q$, a rational witness lies in $F$.

The \emph{model-checking problem} for $\ELQ$ asks whether
$(\varepsilon,q)\models\psi$, for a LINE $\varepsilon$, a point $q$ and a
closed formula $\psi$. Its \emph{data complexity} measures resources as a function of $\varepsilon$
and $q$ for a fixed formula. Its \emph{combined complexity} lets the formula
vary as well. Rational numbers are encoded by a numerator and a positive
denominator in binary. Input size thus accounts for both the number of
endpoints and their bit lengths.

\begin{thm}[Data complexity; cf.~\cite{qel}]\label{thm:datacomplexity}
For every fixed closed formula $\psi$ of $\ELQ$, the problem of deciding
$(\varepsilon,q)\models\psi$ for a given LINE $\varepsilon$ and a given
point $q\in\Rat$ is solvable in time polynomial in the binary size of
$\varepsilon$ and $q$.
\end{thm}
\begin{proof}
Fix $\psi$. The translation in the proof of Theorem~\ref{thm:evidence},
brought into prenex form, has a number $m$ of variables and a quantifier
prefix that depend on $\psi$ only, as do the case distinctions for norms.
Substitute $q$ for $y$, and let $n$ be the binary size of $\varepsilon$ and
$q$. The atomic constraints then determine $O(n)$ hyperplanes of $\Rat^m$,
with coefficients of $O(n)$ bits, the implied constants depending on $\psi$:
the number of constraints grows with the data, while the dimension remains
fixed.

We decompose $\Rat^m$ cylindrically, ordering the coordinates by the quantifier
prefix. This avoids expanding Boolean combinations into disjunctive normal
form. Starting with the innermost coordinate,
we retain the hyperplanes independent of that coordinate and add those
obtained by eliminating it from each pair of hyperplanes that depend on it.
We discard constant forms. The added hyperplanes specify where two bounds on
the eliminated coordinate coincide. One projection turns $g$ hyperplanes into $O(g^2)$, and there are at most $m$
projections. The number of hyperplanes and the bit lengths of their
coefficients therefore remain polynomial in $n$, with exponents depending on
$\psi$. The number of stages is fixed, and each uses rational arithmetic.

We lift the decomposition in reverse order. Over a face of the decomposition
for the first $i$ coordinates, the signs of the projected hyperplanes fix
the order of the bounds on coordinate $i+1$. The graphs of these bounds, and
the open regions between and beyond them, form the faces above it. Choose sample points coordinate by coordinate: a bound, the midpoint of two
consecutive bounds, an extreme bound shifted by $1$, or $0$ if there is no
bound. These points have polynomial bit length. The number of faces at a
stage is polynomial. On each face of the final decomposition all the given constraints
have constant signs, and the projections of two faces are identical or
disjoint. Evaluate the quantifier-free matrix at the sample points and eliminate
the quantifiers from the inside out. For the faces above a given face, an
existential quantifier takes the disjunction of their truth values, and a
universal quantifier takes the conjunction. The matrix already includes the restrictions on length variables. This
decides $\psi$ at $q$ in polynomial time, as in~\cite{kkr1995} for constraint
queries of fixed dimension.
\end{proof}

The bound concerns data complexity. When the formula is part of the input,
the dimension $m$ and the number of projections are no longer fixed. Model
checking is then interreducible in polynomial time with the theory of
$(\Rat,<,+)$, hence PSPACE-hard~\cite{qel}.

\subsection{Satisfiability in dense time}\label{subsec:densesat}
A formula can define a copy of $\Nat$ inside a linear metric monoid.
This lets us transfer satisfiability from discrete time.

\begin{thm}[Satisfiability in linear time]\label{thm:lineartime}
Let $M$ be a linear metric monoid in future mode. Satisfiability for $\ELN$ in
the propositional presentation is reducible to satisfiability for $\EL(M)$.
Hence satisfiability for $\EL(M)$ is undecidable and $\Sigma^1_1$-hard, and
for the dense logics $\ELQ$ and $\EL(\Ratp)$ it is $\Sigma^1_1$-complete.
\end{thm}
\begin{proof}
Let $g$ and $q$ be atomic propositions and $z$ a length variable. Call the
truth value of $q$ at a point its \emph{colour}. Put
\[
\Grid:=g\wedge\forall x\,\Bigl(g\rightarrow
   \bigl(\boldBox_zq\wedge(g\wedge\neg q)_z\bigr)\vee
   \bigl(\boldBox_z\neg q\wedge(g\wedge q)_z\bigr)\Bigr)_x .
\]
Let $\Grid$ hold at a point $a$ under an interpretation $\alpha$ when $z$
denotes $\ell$. If $g$ holds at a point $c\geq a$, the interval $[c,c+\ell)$
has one colour, and $c+\ell$ carries $g$ and has the other colour. Hence $g$ holds at $a+k\ell$ for every $k\in\Nat$, and the colour of
the points of $[a+k\ell,a+(k+1)\ell)$ alternates with $k$. If $g$ held at a
point $b$ strictly between $a+k\ell$ and $a+(k+1)\ell$, all points of
$[b,b+\ell)$ would have the colour of $b$, although $a+(k+1)\ell$ is one of
them. Consequently, $g$ holds at or after $a$ precisely at the points $a+k\ell$.

Now let $\varphi$ be a closed formula of $\ELN$ in the propositional
presentation, with length variables only, in which $g$, $q$ and $z$ do not
occur. By Lemma~\ref{lem:sorts}, a closed formula can be put in this form. Let $\varphi^{z}$ result from $\varphi$ by replacing every numeral
$k\geq1$ by the $k$-fold sum of $z$, deleting the summands $0$ and the
displacements by $0$, and replacing $\exists n\,\psi$, $\forall n\,\psi$,
$\boldBox_t\psi$ and $\mdiamond_t\psi$ by $\exists n\,(g_n\wedge\psi)$,
$\forall n\,(g_n\rightarrow\psi)$, $\boldBox_t(g\rightarrow\psi)$ and
$\mdiamond_t(g\wedge\psi)$. With $\alpha$, $a$ and $\ell$ as above, put
$\alpha'(k):=\alpha(a+k\ell)$. By induction on the subformulas $\psi$ of
$\varphi$, for all positions $k$ and assignments $\nu$ over $\Nat$: $\psi$
holds at $k$ under $\alpha'$ and $\nu$ if and only if $\psi^{z}$ holds at
$a+k\ell$ under $\alpha$ when $z$ denotes $\ell$ and every variable $n$
denotes $\nu(n)\ell$. Indeed, a displacement leads from $a+k\ell$ to a point
$a+k'\ell$; at such a point, $g_n$ holds if and only if $n$ denotes a positive multiple
of $\ell$. The points in a window satisfying $g$ correspond exactly to the
positions in the associated window of $\Nat$. Thus $\varphi$ holds at
position $0$ under $\alpha'$ if $\Grid\wedge\varphi^{z}$ holds at $a$.
Conversely, let $\varphi$ hold at position $0$
under an interpretation $\alpha'$ over $\Nat$ in which $g$ and $q$ hold
nowhere. Choose a length $\ell$ of $M$, let
$\alpha(k\ell):=\alpha'(k)\cup\{g\}$ and $\alpha(b):=\emptyset$ at all other
points $b$, and add $q$ at the points of $[k\ell,(k+1)\ell)$ for even $k$;
then $\Grid\wedge\varphi^{z}$ holds at $0$ when $z$ denotes $\ell$. Thus
$\varphi$ is satisfiable if and only if $\exists z\,(\Grid\wedge\varphi^{z})$
is satisfiable in $\EL(M)$. With Proposition~\ref{prop:presentations},
Theorems~\ref{thm:main} and~\ref{thm:analytic} give undecidability and hardness.
For $M=\Rat$ and $M=\Ratp$, an interpretation is coded by a set of natural
numbers. Satisfaction is first-order definable over $(M,<,+)$ expanded by
the supports of the atomic propositions, as in the proof of
Theorem~\ref{thm:evidence}(i). Hence satisfiability is in $\Sigma^1_1$.
\end{proof}

\noindent
The fragment of the main theorem also transfers, using at most one additional
shared length.

\begin{cor}[Shared lengths in linear time]\label{cor:threelengths}
Let $M$ be a linear metric monoid in future mode. Satisfiability for $\EL(M)$
is $\Sigma^0_1$-hard, and hence undecidable, already for the closed formulas
$\exists z\,\exists N\,\exists n\,\Psi$ and, if $\EL(M)$ has a length
constant, $\exists N\,\exists n\,\Psi$, in which $\Psi$ is built from literals
over six atomic propositions by $\wedge$, $\vee$, displacement and bounded
universality.
\end{cor}
\begin{proof}
We use the notation of the preceding proof. Let $\Psi_\Delta$ be the matrix
of the formula $\Phi_\Delta$ of Section~\ref{sec:pcp} in the propositional
presentation (Corollary~\ref{cor:prefix}). Let $\mu$ be the largest length of
a word of $\Delta$, and let $h$ be the term $n+N+N+(\mu+3)z$. Let $\Grid_h$ result from
$\Grid$ by replacing the quantifier $\forall x$ and the displacement by $x$
with the window $\boldBox_h$. If $\Grid_h$ holds at a point $a$ when $z$
denotes $\ell$, the argument for $\Grid$ shows that, among the points of
$[a,a+h)$, the proposition $g$ holds precisely at the points $a+k\ell$. As
$a+N$ and $a+n$ lie in this interval, the guards $g_N$ and $g_n$ make $N$ and
$n$ multiples of $\ell$. At position $0$, the evaluation of $\Psi_\Delta$ visits
only positions below $n+2N+\mu+3$, the furthest for a mark that $\Step$
places in the following row. The windows of $\Psi_\Delta^{z}$ therefore lie
within $[a,a+h)$, and the induction of the preceding proof applies. Hence
$\exists z\,\exists N\,\exists n\,(\Grid_h\wedge g_N\wedge g_n\wedge\Psi_\Delta^{z})$
is satisfiable in $\EL(M)$ if and only if $\Delta$ has a solution. Its
literals are those of the four atomic propositions of
Corollary~\ref{cor:prefix} and of $g$ and $q$. As $\ell$ was arbitrary, a
length constant may replace $z$.
\end{proof}

\section{Events in Real Time}\label{sec:mtl}

Our last instance combines discrete event positions with real time. A timed
word orders events and records their timestamps. Distinct events may share a
timestamp; points of $\Nat\times\Real$ retain both position and time. The
product admits two kinds of windows. Metric balls bound position and time
jointly; products of future windows bound them independently. Their
covariance properties differ. We give the definitions, identify which general
laws apply, and transfer the main theorem to timed words. The treatment is
limited to what the questions of Section~\ref{sec:discussion} require.

\subsection{Timed words as interpretations}
Throughout this section $M$ is the monoid $\Nat\times\Real$ with
componentwise addition. A point $(i,r)$ consists of a \emph{position} $i$ and
a \emph{time} $r$. The set $\{i\}\times\Real$ is the \emph{column} of
position $i$. Put $M^{+}:=\Nat\times\Realp$ and call $0=(0,0)$ the
\emph{origin}. Let $AP$ be a finite set of atomic propositions
and $\Ev\notin AP$.

\begin{defi}[Timed words~\cite{timed} and event graphs]\label{def:timedwords}
A \emph{timed word} is a sequence $w=(\sigma_i,\tau_i)_{i\in D_w}$ whose
domain $D_w$ is a nonempty initial segment of $\Nat$, with $\sigma_i\subseteq
AP$, $\tau_0=0$ and nondecreasing \emph{timestamps} $\tau_i\in\Realp$,
unbounded if $D_w=\Nat$. The \emph{events} of $w$ are the points
$e_i:=(i,\tau_i)$ for $i\in D_w$; the event $e_0$ is the origin. The
\emph{event graph} of $w$ is the interpretation $\alpha_w$ in which $\Ev$
holds exactly at the events of $w$, and $p\in AP$ holds exactly at the events
$e_i$ with $p\in\sigma_i$.
\end{defi}

\noindent
Atomic propositions hold only at events. A column contains at most one event,
so the position coordinate distinguishes simultaneous events. The map
$w\mapsto\alpha_w$ is injective. For a point variable $x$ fresh for
$\varphi$, $\exists x\,(\Ev_{x}\wedge\varphi_{x})$ holds at $e_i$ exactly when
$\varphi$ holds at some $e_j$ with $j\geq i$.

\subsection{Balls and rectangular windows}

\begin{defi}[The scalar-metric logic $\SM$]\label{def:sm}
Fix a rational $\lambda>0$. The \emph{scalar-metric logic} $\SM$ is the
logic $\EL(M)$ for the metric $d_\lambda$ of Example~\ref{ex:carriers}(iv) on
$M=\Nat\times\Real$, so that $\norm{(i,r)}=\lambda i+|r|$, with length domain
$\Real_{>0}$, in full mode, with $\Prop=AP\cup\{\Ev\}$, with constants for
the points of $\Nat\times\Rat$ and for the lengths in $\Rat_{>0}$, and with
point variables only.
\end{defi}

\noindent
The window of length $\rho$ at the point $v=(i,r)$ is the open ball
$B(v,\rho)$, which consists of the points $(i+k,r+s)$ of $M$ with
$\lambda|k|+|s|<\rho$. It is the intersection of an open diamond in the real
plane with $M$ (Figure~\ref{fig:carriers}(e)): one length bounds the weighted
sum of the distances in position and time. The full mode on $M$ is not
covariant, since $(1,0)$ has no inverse (Lemma~\ref{lem:covariance}(i)). For
$\rho>\lambda$, the ball $B((1,0),\rho)$ meets the column of the origin, but
$B(0,\rho)+(1,0)$ does not. Displacement therefore does not commute with the
bounded modalities of $\SM$ (Proposition~\ref{prop:disp}(iv)), and
Theorem~\ref{thm:adnf} does not apply. Placing a displacement inside or
outside a bounded modality can change the meaning of the formula.

For a point $v=(i,r)$ of $M^{+}$, $k\in\Natp$ and a real $c>0$, the
\emph{rectangular window} with \emph{extents} $k$ and $c$ is
$W(v;k,c):=[i,i+k)\times[r,r+c)$, the product of the future windows of the
two factors (Definition~\ref{def:windows}). The modalities
$\md_{(k,c)}\theta$ and $\mb_{(k,c)}\theta$ hold at $v$ if $\theta$ holds at
some, respectively at every, point of $W(v;k,c)$. Formulas with these
modalities are interpreted over $M^{+}$, and the parameters of their modal
operators are pairs of terms, one for each coordinate.
Rectangular windows are covariant, $W(v+u;k,c)=W(v;k,c)+u$ for all points
$u$ of $M^{+}$, because both factors are in future mode
(Lemma~\ref{lem:covariance}(ii)). The proofs of
Propositions~\ref{prop:reference} and~\ref{prop:disp} and of
Theorem~\ref{thm:adnf} use no other property of the windows. With rectangular windows, displacement therefore commutes with the two
modalities, and formulas admit atomic-displacement normal form. Translating
an interpretation has the same effect as translating the point of evaluation,
as for $\omega$-words (Lemma~\ref{lem:suffix}).

\subsection{Satisfiability over timed words}
The main theorem transfers to timed words.

\begin{cor}[Shared lengths on timed words]\label{cor:timed}
It is $\Sigma^0_1$-hard whether a closed formula
$\exists N\,\exists n\,\Psi$ holds at the origin of the event graph of some
finite timed word, and likewise of some infinite one, with $N,n$ over $\Natp$
and $\Psi$ built from literals over five atomic propositions by $\wedge$,
$\vee$, displacements by $(s,s)$ and modalities $\mb_{(s,s)}$, for positive
affine terms $s$ of $\ELN$.
\end{cor}
\begin{proof}
Let $\Psi_\Delta$ and $\mu$ be as in the proof of
Corollary~\ref{cor:threelengths}, and let $h$ be the term $n+N+N+\mu+3$. Let
$\Psi_\Delta^{\times}$ result from $\Psi_\Delta$ by replacing every
displacement by $s$ with one by $(s,s)$, and every $\boldBox_s\chi$ with
$\mb_{(s,s)}(\neg\Ev\vee\chi)$; put
$\Diag:=\Ev\wedge\mb_{(h,h)}(\neg\Ev\vee\Ev_{(1,1)})$. For given values of
$N$ and $n$, $\Diag$ holds at the origin of $\alpha_w$ exactly when
$e_j=(j,j)$ for all $j\leq h$: if $e_j=(j,j)$ and $j<h$, the event $e_j$
lies in $W(0;h,h)$, so $\Ev$ holds at $(j+1,j+1)$. The events in a window
$W(e_j;k,k)$ with $j+k\leq h$ are then $e_j,\ldots,e_{j+k-1}$. Since the
evaluation of $\Psi_\Delta$ at position $0$ visits only positions below $h$,
induction shows that $\Psi_\Delta^{\times}$ holds at the origin of $\alpha_w$
exactly when $\Psi_\Delta$ holds at position $0$ under the interpretations
$\alpha$ with $\alpha(j)=\sigma_j$ for $j<h$. By Lemmas~\ref{lem:complete}
and~\ref{lem:sound}, $\exists N\,\exists n\,(\Diag\wedge\Psi_\Delta^{\times})$
therefore holds at the origin of the event graph of some finite timed word,
and of some infinite one, if and only if $\Delta$ has a solution.
\end{proof}

\noindent
Here the timestamps $\tau_j=j$ are strictly increasing. Simultaneous events
play no part; undecidability stems from the shared lengths. Satisfiability of
future metric temporal logic ($\MTL$)~\cite{metric,alurh93}, by comparison, is
decidable over finite timed words~\cite{worrell}.

\section{Discussion and Open Problems}\label{sec:discussion}

The logics share a syntax and satisfaction relation, but their laws and
algorithms depend on the underlying metric monoid. Additive decomposition, for
example, is a law of linear time. In $\Real^k$ with $k\geq2$, two balls of
radii $s$ and $t$ do not cover a ball of radius $s+t$, as a comparison of
volumes shows. The composition law of dense time does hold for $\EL(\Real^k)$ in
full mode, by convexity. Satisfiability is
undecidable for the linear-time logics (Theorem~\ref{thm:lineartime})
and also in the plane. In $\EL(\Int^2)$ the tiling problem of
Proposition~\ref{prop:universal} can be expressed directly, since displacements by
$(1,0)$ and $(0,1)$ test the neighbours of a point~\cite{berger1966}.
Over LINEs, by contrast, satisfiability is recursively enumerable, since model
checking is decidable (Theorem~\ref{thm:evidence}). It remains
undecidable by a reduction from finite satisfiability of first-order
logic~\cite{qel}, and is thus $\Sigma^0_1$-complete. Model checking of $\ELN$
over finite words, with quantifiers restricted to the positions of the word,
has polynomial data complexity and PSPACE-complete combined
complexity~\cite{DZ26}.

For data analysis, the satisfaction relation specifies the result an
implementation is required to compute. It gives the operators a common
reading over the time line of a recording, the coordinates of a sequence,
and the plane or volume of an image. Formulas without negation are
monotone in the supports of the atomic propositions~\cite{qel}. Thus an
inner and an outer approximation of an uncertain record bound the answer.

\noindent
These results lead to the following questions.
\begin{enumerate}[leftmargin=2em]
\item \emph{Automata and the expressive power of $\ELN$.} Prove
Conjecture~\ref{conj:proper}, for instance by Ehrenfeucht--Fra\"{\i}ss\'e
games~\cite{DZ26}. Is there an automaton model that recognizes exactly the
$\ELN$-definable languages? A formula in ADNF tests letters at positions that
are affine in its variable values (Theorem~\ref{thm:adnf}). This
suggests automata with registers for these values. Which classes of
$\omega$-automata, for instance Parikh automata~\cite{oregular}, correspond
to natural fragments of $\ELN$?
\item \emph{Metric temporal logic on event graphs.} Future $\MTL$ is
interpreted at the events of a timed word~\cite{furia,how}:
$\chi_1\U_I\chi_2$ holds at $e_i$ if $\chi_2$ holds at some $e_j$ with
$j\geq i$ and $\tau_j-\tau_i$ in the interval $I$, and $\chi_1$ holds at
$e_k$ for $i\leq k<j$. As in Lemma~\ref{lem:until}, a displacement under the
guard $\Ev$ reaches the witness. It remains to select the events before it
and to confine the elapsed time to $I$. Is every formula $\chi$ of future
$\MTL$ \emph{defined uniformly at events} by a closed formula $\varphi$ with
rectangular windows, so that $\chi$ holds at $e_i$ in $w$ exactly when
$(\alpha_w,e_i)\models\varphi$? And in $\SM$, whose balls couple position
and time? A negative answer requires a method for proving undefinability
in $\SM$.
\item \emph{Beyond metric temporal logic.} A shared duration relates delays
to each other, whereas an interval of $\MTL$ compares a delay with
constants. The formula $\exists z\,(\Ev_{(1,z)}\wedge\Ev_{(2,z+z)})$, with $z$ over
$\Realp$, states that the first two delays are equal. Which properties
do shared durations define, and how do the logics over $\Nat\times\Real$
compare with $\MTL$ with past operators and with timed propositional
temporal logic~\cite{alur-henzinger}? Shared lengths preclude decidability
on finite timed words (Corollary~\ref{cor:timed}). Which fragments beyond
$\MTL$ retain it~\cite{worrell}?
\item \emph{Dense time and real lengths.} What is the exact degree of
satisfiability for $\EL(\Realp)$, which is $\Sigma^1_1$-hard by
Theorem~\ref{thm:lineartime}, and for existential formulas in dense time,
whose windows are infinite? Does Proposition~\ref{prop:onevar} extend to
dense time? The upper
bound for $\ELQ$ codes an interpretation by a set of natural numbers and the
values of variables by numbers. A real length can itself code such a
set. Does quantification over real lengths raise the degree
of satisfiability, in linear time or on timed words?
\item \emph{Decidable existential formulas.} When $\boldBox_t$ is applied to
literals only, as in patterns of consecutive phases such as $\exists n\,\exists m\,(\boldBox_np\wedge(\boldBox_mq)_n)$,
satisfiability of existential formulas is again a condition of existential
Presburger arithmetic. The reduction of Section~\ref{sec:pcp} applies
$\boldBox_t$ to formulas with displacements. Which intermediate fragments
have decidable satisfiability, and what is the complexity for the formulas of
Proposition~\ref{prop:onevar}?
\item \emph{Past and subtraction.} How do expressiveness and decidability change when
we admit the past mode or terms $s-t$, over the forward linear metric
monoids? A comparison must specify both the temporal operators and the
admissible distance constants, since their choice matters for real-time
satisfiability~\cite{past,worrell}.
\item \emph{Space.} Which fragments of the spatial logics are decidable, and
which admit effective model checking on finite geometric presentations, such
as finite unions of balls with rational centres and radii? What
is the relation to qualitative spatial logics~\cite{cohn,debaghian}, to
spatial logics of closure spaces~\cite{cllm2016}, and to
mathematical morphology~\cite{serra1982}?
\item \emph{Certified evaluation.} On a LINE, evaluation is exact rational
arithmetic, and a positive answer to an existential specification has a
rational witness (Theorem~\ref{thm:evidence}). Can a proof assistant such as
Lean~\cite{lean4} certify every answer for large collections of records, for
instance by validating witnesses found by an unverified search?
\end{enumerate}

\section{Conclusion}\label{sec:conclusion}
Ensemble Logic makes the parameters of modal operators into terms and lets
quantified variables be shared among them. The main theorem measures what
this sharing costs in discrete time: two existentially quantified lengths
suffice for undecidability when displacement and bounded universality are
the only modal operators. The boundary is sharp: with one
existential quantifier, or without any one of the four constructs of the
matrix, satisfiability is decidable. Undecidability extends to the
linear-time logics of the family. Two shared lengths suffice again in
dense time and on timed words.

Evaluation, the task arising in applications, behaves differently. A
fixed formula is evaluated on a labelled interval ensemble in polynomial
time. Over metric monoids the same semantic clauses serve discrete time,
dense time, space and their products. The differences between these
settings appear as explicit assumptions. Covariance characterizes when
displacement commutes with the bounded modalities; discrete and dense
linear time satisfy different composition laws. The remaining questions concern automata for the discrete logic,
real-time logics on timed words, quantification over real lengths, further
decidable fragments, and the certified evaluation of specifications.

\section*{Acknowledgments}
The author thanks Licong Cui, Manfred Droste, Xubing Hao, Yan Huang and Xiaojin Li
for feedback and illuminating discussions. This work has been supported in part by the U.S.
National Science Foundation Award IIS-2500624 and the National Institutes of Health grants R01AG084236 and
U24AG098157. The views expressed in this paper are those
of the author and do not necessarily reflect those of the funding agencies.

\end{document}